\documentclass[12pt]{article}
\pdfoutput=1 

\usepackage{bbm}
\usepackage{epsf}
\usepackage{amssymb}
\usepackage{amsfonts}
\usepackage{amsmath}
\usepackage{amsthm}
\usepackage{graphicx}
\usepackage{cite}

\usepackage{latexsym}
\usepackage{xcolor}
\usepackage{hyperref}
\usepackage{color}
\usepackage{transparent}
\usepackage{bm}
\usepackage{tikz-cd} 

\makeatletter
\def\@xfootnote[#1]{%
  \protected@xdef\@thefnmark{#1}%
  \@footnotemark\@footnotetext}
\makeatother

\newcommand{\bmat}{\left(\begin{array}}
\newcommand{\emat}{\end{array}\right)}

\def\p{\partial}
\def\a{\alpha}

\def\d{\delta}

\def\om{\omega}

\def\-{\hphantom{-}}

\def\s2{\frac{1}{\sqrt2}}

\def\beq{\begin{equation}}
\def\eeq{\end{equation}}
\def\beqa{\begin{eqnarray}}
\def\eeqa{\end{eqnarray}}

\def\dim{{\rm dim \,}}

\def\Dsl{\,\raise.15ex\hbox{/}\mkern-13.5mu D} %this one can be subscripted

\def\CN {{\cal N}}

\def\CL {{\cal L}}

\def\be{\begin{equation}}
\def\ee{\end{equation}}
\def\bea{\begin{eqnarray}}
\def\eea{\end{eqnarray}}
\def\raw{\rightarrow}

\def\IC{\mathbb{C}}

\def\IZ{\mathbb{Z}}
\def\IR{\mathbb{R}}

\def\a{{\alpha}}

\def\om{{\omega}}

\def\p{{\partial}}

\makeatletter
\newsavebox{\@brx}
\newcommand{\llangle}[1][]{\savebox{\@brx}{\(\m@th{#1\langle}\)}%
  \mathopen{\copy\@brx\kern-0.5\wd\@brx\usebox{\@brx}}}
\newcommand{\rrangle}[1][]{\savebox{\@brx}{\(\m@th{#1\rangle}\)}%
  \mathclose{\copy\@brx\kern-0.5\wd\@brx\usebox{\@brx}}}
\makeatother

\def\sm2{{\mbox{\small 2}}}

\newcommand{\bp}{\begin{pmatrix*}[r]}  
\newcommand{\ep}{\end{pmatrix*}}  
\newcommand{\bpp}{\begin{pmatrix}}  
\newcommand{\epp}{\end{pmatrix}}  
\newcommand{\bcd}{\begin{center}
\begin{tikzcd}}
\newcommand{\ecd}{\end{tikzcd} \end{center}}

\def\cL{\mathcal{L}}

\def\cO{\mathcal{O}}

\def\1{\mathbb{1}}

\def\d{{\rm d}}

\newtheorem{theorem}{Theorem}[section]
\newtheorem{proposition}{Proposition}[section]
\newtheorem{lemma}{Lemma}[section]
\newtheorem{corollary}{Corollary}[section]

\theoremstyle{remark}

\newtheorem{remark}[theorem]{Remark}
\newtheorem{definition}{Definition}[section]

\DeclareMathOperator{\End}{End}
\DeclareMathOperator{\Id}{Id}
\DeclareMathOperator{\Image}{Im}
\DeclareMathOperator{\rk}{rk}
\DeclareMathOperator{\Hom}{Hom}
\DeclareMathOperator{\tr}{tr}
\DeclareMathOperator{\Ext}{Ext}

\begin{document}
\pagestyle{plain}

%----------------------------------------------------------------------%
%  numbering equations with section number
%----------------------------------------------------------------------%
\makeatletter
\@addtoreset{equation}{section}
\makeatother
\renewcommand{\theequation}{\thesection.\arabic{equation}}
%----------------------------------------------------------------------%
%  title page
%----------------------------------------------------------------------%
\pagestyle{empty}
%\vspace*{1.0in}
\rightline{IFT-UAM/CSIC-20-101}
\rightline{ROM2F/2020/02}
%\rightline{\tt hep-th/yymmnnn}
\vspace{0.5cm}
\begin{center}
\Huge{{A vanishing theorem for T-branes}
\\[15mm]}
\normalsize{Fernando Marchesano,$^1$ Ruxandra Moraru,$^{2}$ and Raffaele Savelli$^{3}$ \\[10mm]}
\small{
${}^1$ Instituto de F\'{\i}sica Te\'orica UAM-CSIC, Cantoblanco, 28049 Madrid, Spain \\[2mm] 
${}^2$  Department of Pure Mathematics, University of Waterloo, Canada \\[2mm] 
${}^3$ Dipartimento di Fisica, Universit\`a di Roma “Tor Vergata” \& INFN - Sezione di Roma2 \\
Via della Ricerca Scientifica, I-00133 Roma, Italy
\\[8mm]} 
\small{\bf Abstract} \\[5mm]
\end{center}
\begin{center} 
\begin{minipage}[h]{15.0cm} 

We consider regular polystable Higgs pairs $(E, \phi)$ on compact complex manifolds. We show that a non-trivial Higgs field $\phi \in H^0 ({\rm End} (E) \otimes K_S)$ restricts the Ricci curvature of the manifold, generalising previous results in the literature. In particular $\phi$ must vanish for positive Ricci curvature, while for trivial canonical bundle it must be proportional to the identity. For K\"ahler surfaces, our results provide a new vanishing theorem for solutions to the Vafa--Witten equations. Moreover they constrain supersymmetric 7-brane configurations in F-theory,
giving obstructions to the existence of T-branes, i.e.~solutions with $[\phi, \phi^\dagger] \neq 0$. When non-trivial Higgs fields are allowed, we give a general characterisation of their structure in terms of vector bundle data, which we then illustrate in explicit examples.

\end{minipage}
\end{center}
\newpage
%----------------------------------------------------------------------%
%  Resetting of counters
%----------------------------------------------------------------------%
\setcounter{page}{1}
\pagestyle{plain}
\renewcommand{\thefootnote}{\arabic{footnote}}
\setcounter{footnote}{0}
%----------------------------------------------------------------------%
%  Paper begins
%----------------------------------------------------------------------%

%\end{document}

\tableofcontents

%\newpage

\section{Introduction}
\label{s:intro}

Vafa--Witten (VW) systems are a set of gauge-theoretic equations on a four-manifold $S$, introduced by Vafa and Witten \cite{Vafa:1994tf} in the study of the S-duality conjecture for $\CN = 4$ supersymmetric Yang–Mills theory. A key role in this analysis is played by a vanishing theorem, thanks to which for some choices of $S$ the partition function of topologically twisted $\CN = 4$ Yang–Mills simplifies, and can be computed in terms of  the Euler characteristic of the moduli space of instantons on $S$. In general such a partition function contains valuable information on four-manifold invariants, as studied in cases where the theorem does not apply, see e.g. \cite{Tanaka:2017jom,Tanaka:2017bcw} for recent progress in this direction. 

The VW equations simplify considerably when $S$ is a K\"ahler surface, and can be written in terms of a Higgs pair $(E, \phi)$, with $E$ a Hermitian vector bundle and $\phi$ a Higgs field. In this sense, they can be seen as a higher-dimensional analogue of the Hitchin equations for Higgs bundles on compact Riemann surfaces \cite{Hitchin:1986vp}. However, unlike in the generalisation made in \cite{Simpson}, in the VW case $\phi$ is twisted by the canonical bundle of $S$. 

Interestingly, the same VW equations for K\"ahler surfaces describe supersymmetric configurations of 7-branes in F-theory \cite{Donagi:2008ca,Beasley:2008dc}, which have been proposed to realise models of Grand Unification in string theory, see \cite{Heckman:2010bq,Weigand:2010wm,Wijnholt:2012fx} for reviews. In this setup, the Higgs field $\phi$ contains the transverse geometrical deformations of the 7-brane embedding, and the vanishing results of \cite{Vafa:1994tf} signal cases where a profile for $\phi$ is %such deformations are 
obstructed. Relevant for the analysis of \cite{Vafa:1994tf} are also cases where $\phi \neq 0$ but it has an Abelian profile, in the sense that the commutator $[\phi, \phi^\dagger]$ appearing in the VW equations vanishes. Remarkably the opposite case $[\phi, \phi^\dagger] \neq 0$,  dubbed T-brane in the string theory literature \cite{Cecotti:2010bp}, is particularly interesting from the 7-brane perspective. Indeed, T-brane configurations are a key ingredient to obtain realistic Yukawas in F-theory GUT model building \cite{Hayashi:2009bt,Cecotti:2010bp}, and as such it is an important problem to understand for which surfaces they can be realised.

In this work we study from a general perspective manifolds where T-branes are obstructed. We consider Higgs pairs $(E, \phi)$ with the same twist as in VW systems but, for the sake of generality, we do it on complex manifolds $X$ of any dimension, endowed with a Gauduchon metric so that Higgs polystability is well defined. We prove a vanishing theorem -- Theorem \ref{VT} -- in terms of the properties of the canonical bundle $K_X$. If the degree of $K_X$ is negative then $\phi$ must vanish, and if $K_X$ is trivial then $\phi$ must be proportional to the identity. Restricting our analysis to K\"ahler metrics, we find that T-branes are only allowed when $K_X$ has positive degree. 

Both our strategy and vanishing theorem are different from the ones in \cite{Vafa:1994tf}. Instead, our proof follows the philosophy of the no-go theorem of \cite{Marchesano:2017kke}, which uses the existence of $\phi$ and the stability of $(E, \phi)$ to find an obstruction in terms of $K_X$. Compared to the vanishing results in \cite{Tanaka:2017jom,Marchesano:2017kke} our present analysis is completely general, in the sense that there are no assumptions either on the manifold $X$ or on the vector bundle $E$.

For those manifolds where the Higgs field is not obstructed, one may describe its different components in terms of vector bundle data. We illustrate this procedure for rank-$2$ bundles, namely {\it a)} an extension of two line bundles and {\it b)} its twist by an ideal sheaf. When restricting the construction to K\"ahler surfaces, we are able to interpret our results in terms of 7-brane physics. From this perspective, the allowed components of $\phi$ correspond to the F-flat directions of a D7-brane superpotential in case {\it a)}, and  
of the superpotential of a D7/D3-brane bound state in case {\it b)}. 

The paper is organised as follows. In Section \ref{s:nonsplit} we review the relation of Vafa-Witten systems, T-branes and polystable Higgs pairs, as well as the vanishing results for the Higgs field in the existing literature. We also describe a simple case to which none of these partial results apply. In Section \ref{s:nogo} we derive the main theorem of this paper, the Vanishing Theorem \ref{VT}, as well as some of its consequences. In those cases where Higgs fields are allowed, they should be characterised in terms of vector bundle data, as we illustrate in Section \ref{s:sufficient}, both from a mathematical and a physics perspective. Section \ref{s:examples} describes explicit examples of these constructions, and in Section \ref{s:conclu} we draw our conclusions. 

Several technical details have been relegated to the appendices. Appendix \ref{ap:homalgebra} discusses several homological-algebra facts needed in Section  \ref{s:sufficient}. Appendix \ref{TopNonSp} discusses how the ideal-sheaf twists can be used to construct topologically non-split bundles. Appendix \ref{ap:cohomology} provides the cohomological computations needed for the constructions of Section \ref{s:examples}.

\section{T-branes in Vafa-Witten systems}
\label{s:nonsplit}

Let us consider a Vafa-Witten system \cite{Vafa:1994tf} on a compact K\"ahler surface $S$. Such a system is specified by the curvature $\mathbb{F} = d\mathbb{A} - i\mathbb{A} \wedge \mathbb{A}$ of a Hermitian vector bundle $E$, and by a (2,0)-form Higgs field $\phi$, which also transforms in the adjoint of the structure group $G_S$.\footnote{When formulating the Vafa--Witten system in terms of eqs.\eqref{VW} we fix the unit metric on the bundle $E$, under which $\mathbb{F}$, $\mathbb{A}$ are Hermitian.}  In terms of these two objects the Vafa-Witten equations read
\begin{subequations}
\label{VW}
\begin{align}
\label{Fterm1}
\bar \partial_{\mathbb{A}} \phi = &\, 0\, ,\\
\label{Fterm2}
\mathbb{F}^{(0,2)} =&\, 0\, ,\\
\label{Dterm}
\om \wedge \mathbb{F} +\frac{1}{2} [\phi, \phi^\dagger] =&\, c\, \om^2\, \Id_E  \,,
\end{align}
\end{subequations}
with $\om$ the K\"ahler two-form of $S$ and $c\in \IR$ a constant. In the string theory literature, such a system would arise from a stack of branes wrapped around $S$, hosting a super-Yang-Mills theory with a symmetry group $G_S$ \cite{Donagi:2008ca,Beasley:2008dc}. There $\phi$ contains the transverse geometrical deformations of the brane embedding, and $\mathbb{A}$ the gauge degrees of freedom. Eqs.\eqref{VW} guarantee that the system preserves supersymmetry and minimises its energy. 

In general, solutions to the above equations can be specified in terms of a {\em Higgs pair}
\be
\label{Hpair}
(E, \phi)\, , \quad \quad \phi \in H^0 \left({\rm Hom}(E, E\otimes K_S)\right)\, ,
\ee
with $K_S$ the canonical bundle of $S$. Via a Hitchin-Kobayashi correspondence, proved for this case in \cite{AlvarezConsul:2001um} (see also \cite{Tanaka:2013yji}), solutions to these equations are associated to polystable Higgs pairs.
\begin{definition}
\label{stable}
A Higgs pair $(E, \phi)$ is {\em stable} if for any non-zero coherent subsheaf $\mathcal{S}$ of $E$ such that $\rk(\mathcal{S})<\rk(E)$ and  $\phi(\mathcal{S})\subseteq \mathcal{S}\times K_S$ we have that
\be
\mu(\mathcal{S}) < \mu (E)\, ,
\label{stability}
\ee
where the {\em slope} of a torsion-free sheaf $\mathcal{E}$ is defined as $\mu(\mathcal{E}) = \deg(\mathcal{E})/\rk(\mathcal{E})$, with\footnote{We can extend these definitions to compact complex manifods $X$ of any dimension $n$, by replacing $S$ by $X$ in \eqref{Hpair} and below, and by defining the degree of a sheaf as
\[ \deg(\mathcal{E}) = \int_X c_1(\mathcal{E}) \wedge \omega^{n-1}, \]
where $\omega$ is the fundamental form of a K\"ahler metric on $X$, or Gauduchon if $X$ is non-K\"ahler. }
\[ \deg(\mathcal{E}) = \int_S c_1(\mathcal{E}) \wedge \omega. \]
\end{definition}
Replacing $<$ by $\leq$ in \eqref{stability} defines semistability, while polystable Higgs pairs amount to a direct sum of stable pairs with the same slope.\footnote{Note that subsheaves $\mathcal{S}$ of $E$ such that ${\rm rk}(\mathcal{S})={\rm rk}(E)$ always have the property that $\mu(\mathcal{S})\leq\mu(E)$. This is because the quotient $E/\mathcal{S}$ is a torsion sheaf (as it has zero rank), and the first Chern class of a torsion sheaf is dual to a positive Cartier divisor (and to the zero divisor if the support of $E/\mathcal{S}$ has codimension strictly higher than one).}

A particularly interesting subset of solutions to \eqref{VW} are those in which $ [\phi, \phi^\dagger] \neq 0$, dubbed {\em T-branes} in the string theory literature \cite{Donagi:2003hh,Cecotti:2010bp,Donagi:2011jy}. For those, as the corresponding $\phi$ is not diagonalisable, the criterion for Higgs bundle stability will be less restrictive than plain polystability for the bundle $E$, allowing for more general constructions.

A standard example where this happens is the original construction of Hitchin \cite{Hitchin:1986vp} generalised and adapted to the Vafa-Witten system. To construct it, let us consider a rank-$2$ bundle $E$ on $S$ of the split type
\be\label{split}
E = \mathcal{L}_1\oplus\mathcal{L}_2
\ee
and fix its complex structure to have $\mathbb{A}^{(0,1)} =0$. This allows to express the Higgs field $\phi$ as a two-by-two matrix such that
\be\label{NilpPhi}
\phi=\left(\begin{array}{cc}
0&m\\ 0 &0
\end{array}\right)\, ,
\ee
with $m \in H^{2,0}(S,\CL_1\otimes \CL_2^{-1})$. Higgs stability amounts to imposing that $\deg(\mathcal{L}_1) < \deg(\mathcal{L}_2)$, which is easily satisfied. Explicit solutions of this sort were constructed in \cite{Marchesano:2017kke}, where it was also pointed out a general obstruction to embedding them in arbitrary K\"ahler surfaces. Indeed, since $m$ can be seen as a section $m \in H^0 (S, \CL_1\otimes \CL_2^{-1} \otimes K_S)$ its existence implies 
\be\label{Esplit}
\exists \,m \quad \Rightarrow \quad  \deg(\mathcal{L}_1) - \deg(\mathcal{L}_2) + \deg(K_S) \geq 0\, ,
\ee
which together with the Higgs-stability condition requires that
\be\label{nogo}
\deg(K_S) > 0\, .
\ee
In particular, solutions of this type cannot be realised in surfaces of positive curvature, analogously to the case of Higgs bundles on Riemann surfaces \cite{Hitchin:1986vp}. It was shown in \cite{Marchesano:2017kke} that this no-go result applies to more general T-brane configurations, with the restriction that $E$ must be a rank-$n$ bundle of the split type. 

The condition $\deg(K_S) \geq 0$ is simple to show for non-nilpotent Higgs fields, since then at least one Casimir of $\phi$ of degree $p$ is non-vanishing. As such a Casimir can be seen as a section of $H^0(S, K_S^p)$, its existence leads to the requirement $\deg(K_S) \geq 0$.  Therefore, showing that \eqref{nogo} is a necessary condition for nilpotent T-branes leads to the following statement
\begin{equation}\label{VT-}
\boxed{\phi =0 \text{  whenever  } \deg(K_S) < 0}
\end{equation}
An interesting result in this direction has been recently obtained in \cite{Tanaka:2017jom,Tanaka:2017bcw}, in which the case of $\IC^*$-fixed Higgs pairs in projective surfaces was studied in detail. Under those assumptions, it was shown that $\phi =0$ if \eqref{nogo} is not met.

The case $\deg(K_S) = 0$ needs to be treated separately, because the above results do not exclude T-branes in general. Nevertheless, we sill show that  for any choice of bundle metric
\begin{equation}\label{VT0}
 \boxed{[\phi, \phi^\dagger] =0 \text{  for  } \deg(K_S) = 0}
\end{equation}

In this work we provide a general proof that \eqref{nogo} is a necessary condition for Vafa-Witten systems with $[\phi, \phi^\dagger] \neq 0$, extending the aforementioned results. More precisely, we will show the statements \eqref{VT-} and \eqref{VT0}. 
The proof will follow a similar strategy to the one used in \cite{Marchesano:2017kke}. For the sake of clarity we will first illustrate our approach in a simple class of Higgs pairs that falls outside the cases covered by \cite{Tanaka:2017jom,Tanaka:2017bcw,Marchesano:2017kke}. Then, in the next section we will provide the general proof.

\subsubsection*{A simple non-trivial example}
\label{ss:nonsplit}

Let us consider a Higgs pair $(E,\phi)$ such that it is Higgs (poly)stable but $E$ itself is \emph{not} (poly)stable as a bundle. Then, by the Hitchin-Kobayashi correspondence for Higgs bundles, we are guaranteed to find a unique solution that must be of T-brane type. This is because, if this solution were such that $[\phi, \phi^\dagger] = 0$, then by the standard Hitchin-Kobayashi correspondence, the bundle $E$ would itself be (poly)stable, contradicting the hypothesis. That is, we have the following

\begin{proposition}
	If $(E,\phi)$ is Higgs stable, but $E$ is not polystable, then $[\phi,\phi^\dagger] \neq 0$. %(that is, there exists a metric $h$ on $E$ that satisfies the Vafa-Witten equations and with respect to which $[\phi,\phi^\dagger] \neq 0$).
\end{proposition}

On the other hand, if $E$ is (poly)stable, we cannot exclude the existence of two different solutions to the Vafa-Witten equations. One with a given hermitian metric and such that $[\phi, \phi^\dagger] = 0$, which must always exist, and a second one with a different hermitian metric and with  $[\phi, \phi^\dagger] \neq 0$. The example that we discuss below realises this circumstance.\footnote{We would like to thank R.~P.~Thomas for suggesting this example to us.}

Let $X$ be a complex manifold of any dimension, and let $\mathcal{L}_1,\mathcal{L}_2$ be two line bundles over it. Consider the following non-trivial extension
\begin{equation}\label{ext}
0\xrightarrow{}\mathcal{L}_1\xrightarrow{i} E\xrightarrow{q} \mathcal{L}_2\xrightarrow{}0\,,
\end{equation}
where $i$ is an embedding. Let us in addition consider a map $\varphi:\mathcal{L}_2\to \mathcal{L}_1\otimes K_X$. One may construct a Higgs pair $(E,\phi)$, by taking the Higgs field to be $\phi=i\,\circ\,\varphi\,\circ q:E\to E\otimes K_X$. Note that since \eqref{ext} is a complex, then $\phi^2=0$. Moreover, for the same reason, $\phi|_{\mathcal{L}_1}=\phi\,\circ\, i=0$, which implies that the only proper subbundle of $E$, i.e.~$\mathcal{L}_1$, is left invariant by $\phi$. In other words, requiring that $(E,\phi)$ is a stable Higgs pair is equivalent to requiring the standard slope-stability for $E$. 

Let us assume that $E$ is stable, that is $\mu(\mathcal{L}_1)<\mu(E)$. This in particular implies that $H^0(X,{\rm End}(E))=\mathbb{C}$, which makes it impossible to undo a $\mathbb{C}^*$ action on $\phi$ by a $\mathbb{C}^*$ action on $E$. This would only be possible if the extension \eqref{ext} is such that $E=\mathcal{L}_1\oplus \mathcal{L}_2$, because then an endomorphism of the type ${\tiny \left(\begin{array}{cc}1&0\\0&\lambda\end{array}\right)}$ with $\lambda\in\mathbb{C}^*$ would undo the action $\phi\to\lambda\phi$. In general, however, the extension \eqref{ext} leads to nilpotent, non-$\mathbb{C}^*$-fixed Higgs pairs, and falls outside the analysis of \cite{Tanaka:2017jom,Tanaka:2017bcw}.

In order to apply this construction to the Vafa-Witten equations, let us take $X$ to be a compact K\"ahler manifold $S$. Then, because the whole construction is based on two line bundles $\CL_1$ and $\CL_2$, one may formulate the derivation of the vanishing theorem for this class of examples in a pedestrian way, very similar to the reasoning for the split bundle case. Indeed, one again obtains that the stability condition implies $\deg(\mathcal{L}_1) < \deg(\mathcal{L}_2)$, while the existence of $\varphi$ requires the positivity condition \eqref{Esplit}. Putting both together one obtains
\be\label{inext}
\deg(K_S) >  \deg(\mathcal{L}_1) - \deg(\mathcal{L}_2) + \deg(K_S) \geq 0\, ,
\ee
completing the proof of the vanishing theorem for this class of Higgs pairs.

The class of examples constructed by means of the non-trivial extension \eqref{ext} gives rise to vector bundles $E$ which, \emph{topologically}, are still equivalent to $E\simeq \mathcal{L}_1\oplus\mathcal{L}_2$. This means that they can be continuously connected to a split form by deforming their holomorphic structure. If dim$X>1$, however, one can have vector bundles which are topologically inequivalent to a sum of line bundles, that even more so calls for a general proof of the vanishing theorem. One particular way to build such cases is to twist the sequence \eqref{ext} by an ideal sheaf supported on a codimension-2 locus. We will give the details of this construction in Section \ref{s:sufficient}, and proceed now with proving the theorem in full generality.

\section{The vanishing Theorem}
\label{s:nogo}

Let us now turn to the general proof of the vanishing theorem. Our strategy will be analogous to the cases discussed in the previous section. We will use the stability of the Higgs pair $(E,\phi)$ and the non-triviality of the Higgs field $\phi$ to constrain the canonical bundle $K_X$ of the complex manifold $X$. Our discussion extends trivially to polystable Higgs bundles, which can be seen as direct sums of stable Higgs bundles with the same slope. 

We first prove a preliminary result:

\begin{lemma}\label{slope}
Consider the short exact sequence of torsion-free sheaves
\[ 0 \rightarrow \mathcal{P} \rightarrow \mathcal{E} \rightarrow \mathcal{Q} \rightarrow 0\, . \]
\begin{itemize}
\item[{\it i)}]
If $\mu(\mathcal{P}) =  \mu(\mathcal{E})$, then $\mu(\mathcal{P}) = \mu(\mathcal{E}) =  \mu(\mathcal{Q})$. 
\item[{\it ii)}]
If $\mu(\mathcal{P}) <  \mu(\mathcal{E})$, then $\mu(\mathcal{P}) < \mu(\mathcal{E}) <  \mu(\mathcal{Q})$.
\end{itemize}
\end{lemma}

\begin{proof}
Applying the definitions of degree and slope of torsion-free sheaves we have that
 $\deg(\mathcal{E}) = \deg(\mathcal{P}) + \deg(\mathcal{Q})$ 
and that $\rk(\mathcal{E}) = \rk(\mathcal{P}) + \rk(\mathcal{Q})$. Moreover, we have that
\begin{equation}\label{star}
\rk(\mathcal{P})(\mu(\mathcal{E}) - \mu(\mathcal{P})) + \rk(\mathcal{Q})(\mu(\mathcal{E}) - \mu(\mathcal{Q})) = 0\,.
\end{equation}
Now, because  $\mathcal{P}$ and $\mathcal{Q}$ are torsion-free sheaves, $\rk(\mathcal{P}), \rk(\mathcal{Q}) > 0$. Using  \eqref{star}, we have that if $\mu(\mathcal{P}) =  \mu(\mathcal{E})$, then $\rk(\mathcal{Q})(\mu(\mathcal{E}) - \mu(\mathcal{Q})) = 0$, 
implying that $\mu(\mathcal{E}) =  \mu(\mathcal{Q})$. 
If on the contrary we assume that $\mu(\mathcal{P}) <  \mu(\mathcal{E})$. Then, by \eqref{star},
\[  \rk(\mathcal{Q})(\mu(\mathcal{E}) - \mu(\mathcal{Q})) = - \rk(\mathcal{P})(\mu(\mathcal{E}) - \mu(\mathcal{P})) < 0\]
implying that $\mu(\mathcal{E}) <  \mu(\mathcal{Q})$.
\end{proof}

With this result, we now prove the statement \eqref{VT-} for a general complex manifold.

\begin{proposition}\label{VTK-}
The existence of a stable Higgs pair $(E,\phi)$ on a complex manifold $X$ with $\phi \neq 0$ implies that $\deg K_X \geq 0$.
\end{proposition}

\begin{proof}
 Note that the Higgs field $\phi: C^\infty(E) \rightarrow C^\infty(E \otimes K_X)$ of $E$ induces a Higgs field $\phi' = \phi \otimes \Id_{K_X} : C^\infty(E \otimes K_X) \rightarrow C^\infty((E \otimes K_X) \otimes K_X)$ of $E \otimes K_X$. Moreover, $P$ is a $\phi$-invariant subsheaf of $E$ if and only if $P \otimes K_X$ is a $\phi'$-invariant subsheaf of $E \otimes K_X$. This then implies that $(E,\phi)$ is stable if and only if $(E \otimes K_X,\phi')$ is stable. Since we are assuming that $(E,\phi)$ is stable, so is $(E \otimes K_X,\phi')$. Moreover, by assumption $\ker \phi$ is a $\phi$-invariant subsheaf of $E$ and $\Image \phi$ is a non-vanishing $\phi'$-invariant subsheaf of $E \otimes K_X$. Note that $\rk(E)=\rk(E \otimes K_X)$, and that
\begin{equation}\label{slopeq}
\mu(E \otimes K_X)=\mu(E)+\deg(K_X)\,.
\end{equation}

Let us first assume that $\mu(\Image \phi) = \mu(E \otimes K_X)$. Because of the stability of $(E \otimes K_X,\phi')$, this can only be true if $\rk(\Image \phi)=\rk(E \otimes K_X)=\rk(E)$. Since $\Image\phi\simeq E/\ker\phi$ and $\ker\phi$ is free of torsion, this implies that $\ker\phi=\{0\}$, and thus $\phi$ is injective. Hence we have that $E \simeq \Image \phi$ and so $\mu(E) = \mu(\Image \phi)$, implying by \eqref{slopeq} that $\deg(K_X)=0$. 

If on the contrary $\mu(\Image \phi) \neq \mu(E \otimes K_X)$, then necessarily $\mu(\Image \phi)<\mu(E \otimes K_X)$, again by stability of $(E \otimes K_X,\phi')$ and $\phi'$-invariance of $\Image \phi$. Given the short exact sequence
\be\label{KerImSeq}
 0 \rightarrow \ker \phi \rightarrow E \rightarrow E/\ker \phi \simeq \Image \phi \rightarrow 0\,,
 \ee
we consider two cases. If $\ker \phi = 0$, then $E \simeq \Image \phi$ and $\mu(E) = \mu(\Image \phi)$.
Whereas if $\ker \phi \neq 0$ then, it is a proper (because $\phi$ is non-identically vanishing) $\phi$-invariant subsheaf of $E$ and, by stability of $(E,\phi)$, $\mu(\ker \phi) < \mu(E)$. Applying Lemma \ref{slope} to the above sequence we have that $\mu(E) <  \mu(\Image \phi)$. 
Putting all this together we obtain
\[ \mu(E) \leq \mu(\Image \phi) < \mu(E) + \deg K_X\,, \]
where the equality holds iff $\phi$ is injective. As a result  $\deg K_X > 0$.
\end{proof}

In fact, one can say more about the case of degree-zero canonical bundle by showing the following:
\begin{corollary}\label{iso}
For $(E,\phi)$ stable with $\phi \neq 0$, we have $\deg K_X=0$ if and only if $\phi$ is an isomorphism.
\end{corollary}
\begin{proof}
If $\phi$ is an isomorphism then $E\simeq E\otimes K_X$, and by Eq. \eqref{slopeq} it follows that $\deg K_X=0$. To prove the other direction note that when $\deg K_X=0$, then $\mu(\Image \phi) = \mu(E \otimes K_X)$, as follows from the proof of Proposition \ref{VTK-}. Therefore, by stability, $\rk(\Image \phi)=\rk(E \otimes K_X)$. Nevertheless, one must still show that $\phi$ has a vanishing cokernel. For this purpose, consider the injective map between line bundles $\det\phi:\det E\to\det E\otimes K_X$. This map is also surjective, because the torsion sheaf $\det E\otimes K_X/\Image\det\phi$ has both vanishing rank and first Chern class.\footnote{We are implicitly using here that $\det\Image\phi\simeq\Image\det\phi$, but this is clearly true when $\phi$ is injective. Indeed, on the one hand injectivity of $\phi$ implies $\Image\phi\simeq E$ and hence $\det\Image\phi\simeq\det E$. On the other hand, injectivity of $\det\phi$ implies that $\Image\det\phi\simeq\det E$.} Therefore $\det\phi$ is an isomorphism, implying that $\phi$ is as well \cite{takemoto}.
\end{proof}

The structure of the Higgs field in the case of zero-degree canonical bundle is very simple, and leads to a trivial commutator $[\phi,\phi^\dagger]$ for any choice of bundle metric. To see this, notice that in the particular case of trivial $K_X$, a Higgs pair is given by $(E,\varphi)$ with $\varphi\in{\rm End}(E)$. One may in fact consider such a pair, which we will dub {\it untwisted} Higgs pair, for a general complex surface $X$, and define slope stability analogously to Definition \ref{stable}. We can then show the following result, which parallels Schur's Lemma:

\begin{lemma}\label{lambdaId}
Let $(E,\varphi)$ with $\varphi\in{\rm End}(E)$ a stable untwisted pair. Then, $\varphi=\lambda I$ with $\lambda\in\mathbb{C}$.
\end{lemma}
\begin{proof}
Suppose $\varphi$ is not identically vanishing, so that $\{0\}\neq\Image\varphi\subseteq E$. On the one hand, since $\Image\varphi$ is a $\varphi$-invariant subsheaf of $E$ and the pair $(E,\varphi)$ is stable, $\mu(\Image\varphi)\leq\mu(E)$ and the equality holds only if $\rk(\Image\varphi)=\rk(E)$. On the other hand, since $\Image\varphi\simeq E/\ker\varphi$, independently of what $\ker\varphi$ may be, we have that $\mu(E)\leq\mu(\Image\varphi)$. This implies that $\mu(\Image\varphi)=\mu(E)$, and thus that $\rk(\Image\varphi)=\rk(E)$, meaning that $\varphi$ is injective. Since the domain of $\varphi$ coincides with its codomain, injectivity implies surjectivity, and hence $\varphi$ is an automorphism. Moreover, being $\mathbb{C}$ algebraically closed, there exists $\lambda\in\mathbb{C}^*$ such that $\ker(\lambda I-\varphi)\neq\{0\}$. Also, $(E,\varphi)$ is a stable pair if and only if  $(E,\lambda I-\varphi)$ is a stable pair, because for any subsheaf $\mathcal{P}$ of $E$, we have $\varphi(\mathcal{P})\subseteq \mathcal{P}$ if and only if $(\lambda I-\varphi)(\mathcal{P})\subseteq \mathcal{P}$. Therefore, applying again the reasoning above to the map $\lambda I-\varphi:E\to E$, we conclude that $\varphi-\lambda I\equiv0$, since it cannot be an automorphism.
\end{proof}

%In order to apply all these results to Vafa-Witten systems, we may now consider the particular case where $X$ is a K\"ahler surface $S$. We have the following important statement:

Let us now focus on K\"ahler manifolds. From the above result we derive the following:

\begin{proposition}\label{Deg0}
Let $(E,\phi)$ be a stable Higgs pair on a K\"ahler manifold $X$ with $\deg K_X = 0$. Then $[\phi,\phi^\dagger] = 0$ for any choice of bundle metric.
\end{proposition}
\begin{proof}
Corollary \ref{iso} implies that $\phi$ is either zero or an isomorphism. The vanishing of the commutator for the case where $X$ is a Calabi-Yau CY (i.e.~$K_X\simeq\mathcal{O}_{\rm CY}$) follows from Lemma \ref{lambdaId}. All other cases with zero-degree canonical bundle are such that $c_1(K_X)$ is a torsion class of $H^2(X, \IZ)$, and hence that $K_X^r\simeq\mathcal{O}_X$ for some integer $r>1$. We may then define a $r:1$ covering map $\pi \colon {\rm CY} \to X$, with $\pi^* K_X = \mathcal{O}_{\rm CY}$. Just like for bundles \cite{kempf}, by the Hitchin-Kobayashi correspondence \cite{AlvarezConsul:2001um} $\pi^*$ preserves polystability of Higgs pairs and sends $(E, \phi)$ to a slope-polystable Higgs pair $(\pi^*E, \pi^*\phi)$ in CY.\footnote{For instance, for the case of a surface $S$, $\pi^*$ pulls-back the Vafa-Witten system and sends solutions to eqs.\eqref{VW} in $S$ to solutions in $K3$. Then, by the Hitchin-Kobayashi correspondence $(\pi^*E, \pi^*\phi)$ must be a slope-polystable Higgs pair in $K3$.} Assuming that $\phi \neq 0$ in $X$ and applying our above result for CY we find that $\pi^*\phi$ must be block-diagonal, with each block proportional to the identity. Therefore $\pi^*([\phi,\phi^\dagger]) =[\pi^*\phi, (\pi^*\phi)^\dag)] =  0$ pointwise, and so we must necessarily have that $[\phi,\phi^\dagger]=0$.
\end{proof}

Notice that when $K_X$ is trivial this result applies to any complex manifold $X$. In fact, we would expect Proposition \ref{Deg0} to also hold for non-K\"ahler, complex manifolds with a Gauduchon metric and $\deg K_X = 0$. This would be guaranteed by the appropriate generalisation of the Hitchin-Kobayashi correspondence to Gauduchon metrics \cite{LubkeTeleman}.

Putting together Propositions \ref{VTK-} and \ref{Deg0}, we obtain the vanishing theorem:

\begin{theorem}\label{VT}
A polystable Higgs pair $(E,\phi)$ on a complex manifold $X$ with $\phi\neq0$ implies that $\deg K_X \geq 0$. If $K_X$ is trivial then $[\phi,\phi^\dagger] = 0$ for any choice of bundle metric.
\end{theorem}

Finally, restricting to the case of T-branes in Vafa-Witten systems we obtain:

\begin{corollary}\label{VTC}
A polystable Higgs pair $(E,\phi)$ on a K\"ahler surface $S$ with $[\phi,\phi^\dagger] \neq 0$ implies that $\deg K_S > 0$. In particular, $S$ has to be properly elliptic or of general type.
\end{corollary}

%Also argument about non-trivial commutator

\subsubsection*{Topological restrictions on Higgs pairs}

Throughout the proof of Proposition \ref{VTK-} we not only have obtained that $\deg K_S \geq 0$ for stable Higgs pairs $(E,\phi)$, but also that the inequalities
\begin{equation}\label{ineq}
\boxed{\mu(E) \leq \mu(\Image \phi) \leq \mu(E) + \deg K_S}
\end{equation}
must always be true for $\phi\neq 0$. Here the first equality becomes strict whenever $\phi$ is not injective -- like in the case of nilpotent Higgs fields -- and the second one whenever $\phi$ is not surjective -- like for any sort of T-brane. % with the second inequality becoming strict for T-branes, and the first one for the particular case of nilpotent Higgs fields. 
Even in a surface $S$ with $\deg K_S > 0$, the condition \eqref{ineq} will restrict which Higgs pairs can be constructed on it.

Indeed, let us for simplicity consider the case where $\det E\simeq\mathcal{O}$. Then these inequalities are equivalent to
\begin{equation}\label{ineqdet}
0\leq\deg(\det\Image\phi)\leq\rk(\Image\phi)\deg(K_S)\,,
\end{equation}
%which can be translated into a hierarchy of areas of holomorphic curves in $S$. 
This is reminiscent of the inequalities obtained in \eqref{inext} for the simple T-brane example based on the extension \eqref{ext}, or to those obtained in \cite{Marchesano:2017kke} for the case of split bundles. However, they are not fully equivalent. Indeed, taking such a nilpotent example and considering the case where $\deg(\mathcal{L}_1) + \deg(\mathcal{L}_2)= 0$ one obtains that \eqref{ineqdet} becomes
\begin{equation}\label{ineqTex}
0< \deg(\mathcal{L}_2) < \deg(K_S)\,.
\end{equation}
This condition should be supplemented with the inequality $2 \deg(\mathcal{L}_2) \leq \deg(K_S)$ in order to arrive to \eqref{inext} applied to this case. This extra condition describes in a more precise manner which line bundles $\mathcal{L}_2$ are allowed given $\deg(K_S)$. In this sense \eqref{ineqTex} only contains partial information, although it is still stronger than the stability condition $0< \deg(\mathcal{L}_2)$: strong enough to require that $\deg(K_S) > 0$.

\section{Structure of the Higgs field} 
\label{s:sufficient}

Theorem \ref{VT} gives a necessary condition for a complex manifold $X$ to allow for a non-trivial Higgs field, assuming Higgs bundle polystability. But even if $\deg K_X \geq 0$ is satisfied, it is in general not clear whether $\phi$ can actually be non-trivial, except for the components proportional to the identity, which are in one-to-one correspondence with sections of $K_X$. In the following we would like to give sufficient conditions for the existence of non-trivial Higgs fields, with particular emphasis on the traceless data of $\phi$.  We will focus our discussion on the particular rank-2 example analysed in section \ref{s:nonsplit} and on a simple generalisation thereof, for which we will show how to construct $\phi$ in terms of the bundle data. We will then give a physical interpretation of our results, which permits to extend this picture to other settings.

\subsubsection*{Simple rank-2 example}

Let us consider again the extension of two line bundles $\mathcal{L}_1,\mathcal{L}_2$ over a complex manifold $X$, as discussed in section \ref{s:nonsplit}, namely\footnote{To make formulas shorter, in this section we are going to omit the zeroes in the exact sequences.}
\begin{equation}\label{ext2}
\mathcal{L}_1\xrightarrow{i} E\xrightarrow{q} \mathcal{L}_2 \,,
\end{equation}
where $i$ is an embedding. The extension class defining $E$ is an element of the group $\Ext^1(\mathcal{L}_2,\mathcal{L}_1)\simeq H^1(\mathcal{L}_2^{-1}\otimes\mathcal{L}_1)$. Following a standard construction (see e.g.~\cite{Harder}), we may build the following grid of long exact sequences in cohomology:
\begin{figure}[h]
\vspace*{.25cm}
\small \hspace*{-.25cm}
\begin{tikzcd}[row sep=normal,column sep=small]
%& 0 \dar & 0 \dar & 0 \dar & 0 \dar\\ 
%0 \rar & 
\Ext^0(\mathcal{L}_2,\mathcal{L}_1\otimes K_X) \dar{\tilde{i}_*} \rar{q^*} & \Ext^0 (E,\mathcal{L}_1\otimes K_X) \rar{i^*} \dar{\tilde{i}_*} & \Ext^0 (\mathcal{L}_1,\mathcal{L}_1\otimes K_X) \dar{\tilde{i}_*} \rar{\Lambda_3} & \Ext^1 (\mathcal{L}_2,\mathcal{L}_1\otimes K_X) \rar\dar{\tilde{i}_*}&\cdots  \\
%0 \rar & 
\Ext^0(\mathcal{L}_2, E\otimes K_X) \rar{q^*} \dar{\tilde{q}_*} & \Ext^0(E,E\otimes K_X) \rar{i^*} \dar{\tilde{q}_*}& \Ext^0(\mathcal{L}_1,E\otimes K_X)  \rar{}\dar{\tilde{q}_*} & \Ext^1(\mathcal{L}_2,E\otimes K_X)\rar{}\dar{\tilde{q}_*}&\cdots\\
%0 \rar & 
\Ext^0(\mathcal{L}_2,\mathcal{L}_2\otimes K_X) \rar{q^*}\dar{}  & \Ext^0(E,\mathcal{L}_2\otimes K_X) \rar{i^*} \dar{} & \Ext^0(\mathcal{L}_1,\mathcal{L}_2\otimes K_X)  \rar{\Lambda_+} \dar{} & \Ext^1(\mathcal{L}_2,\mathcal{L}_2\otimes K_X)\rar \dar&\cdots\\
%&
\vdots&\vdots&\vdots&\vdots
\end{tikzcd}
\caption{Grid of long exact sequences in cohomology associated to the sequence \eqref{ext2}.}
\label{grid}
\end{figure}

\noindent
where $\tilde{i}\equiv i\otimes {\rm Id}_{K_X}$, $\tilde{q}\equiv q\otimes {\rm Id}_{K_X}$, and the maps $\Lambda_3$, $\Lambda_+$ will be discussed below. Since all the sheaves are locally free here, we simply have 
\begin{equation}
\Ext^i(\cdot,\cdot)\simeq H^i({\rm Hom}(\cdot,\cdot)) := H^i(X,\Hom(\cdot,\cdot)).
\end{equation} 
%We will be using the following notation:
%
%\[ v\in\Ext^0(\mathcal{L}_1,\mathcal{L}_1\otimes K_X)  \simeq H^0(K_X)\,, \]
%\[  m\in\Ext^0(\mathcal{L}_2,\mathcal{L}_1\otimes K_X)\simeq H^0(\mathcal{L}_2^{-1} \otimes \mathcal{L}_1\otimes K_X)\,, \] 
%\[\quad p\in\Ext^0(\mathcal{L}_1,\mathcal{L}_2\otimes K_X)\simeq H^0(\mathcal{L}_1^{-1} \otimes \mathcal{L}_2\otimes K_X)\,,\]
%\[\quad \a \in\Ext^0(E, \mathcal{L}_1\otimes K_X)\simeq H^0 \left( \Hom (E, \mathcal{L}_1 \otimes K_X)\right)\,.\]
%
Moreover, for all $\a\in \Ext^0(E,\mathcal{L}_1 \otimes K_X) = H^0 \left( \Hom (E, \mathcal{L}_1 \otimes K_X)\right)$, we have $i^*(\alpha)=\alpha\circ i$ and $\tilde{i}_*(\a) = (i \otimes \Id_{K_X}) \circ \a$. Similar definitions apply for $q^*$ and $\tilde{q}_*$.
In the following we would like to show how the global sections $\phi \in \Ext^0(E,E\otimes K_X) \simeq H^0 \left( \End (E) \otimes K_X\right)$ can be constructed in terms of the other elements of the grid. We will use the identification $\Hom(\mathcal{E}_1,\mathcal{E}_2) \simeq \mathcal{E}_1^* \otimes \mathcal{E}_2$ and the notation $H^i(\mathcal{E}) := H^i(X,\mathcal{E})$ throughout, for all sheaves $\mathcal{E}$, $\mathcal{E}_1$ and $\mathcal{E}_2$ on $X$.

The traceless Higgs fields are sections of $\End_0 (E) \otimes K_X$, where $\End_0 (E)$ is the sheaf of traceless endomorphisms of $E$.
Note that
\be \End (E) = \mathcal{O}_X \oplus \End_0 (E)\,, \ee
where $\mathcal{O}_X$ corresponds to multiples of $\Id_E$. Consequently
\be H^0 (\End(E) \otimes K_X) = H^0(K_X) \oplus H^0(\End_0(E) \otimes K_X)\, ,\ee
with elements of $H^0(K_X)$ being of the form
$\Id_E \otimes s$
for some global section $s$ of $K_X$. Moreover,
\be
 \tilde{q}_*(\Id_E \otimes s) = q \otimes s \neq 0
\ee
whenever $s \neq 0$. This implies that multiples of the identity inject into $H^0(\Hom(E,\cL_2 \otimes K_X))$ under $\tilde{q}_*$, and thus the image under $\tilde{i}_*$ of every non-zero element of $H^0(\Hom(E,\cL_1 \otimes K_X))$ has a non-zero traceless part, i.e.
\be 0 \neq \left(\tilde{i}_*(\alpha) - \frac{1}{2}\Id_E \otimes \tr(\tilde{i}_*(\alpha)) \right) \in H^0(\End_0(E) \otimes K_X)\,,\ee
for every $0 \neq \alpha \in H^0(\Hom(E, \cL_1 \otimes K_X))$. As a consequence we have the following:
\begin{lemma}
	If $H^0(\Hom(E,\cL_1 \otimes K_X)) \neq 0$, then $E$ admits non-zero traceless Higgs fields.
\end{lemma}
\medskip
Let us describe $H^0(\Hom(E,\cL_1 \otimes K_X))$ in detail, in terms of the first row of the grid in figure \ref{grid}: 
\begin{equation}\label{sections-first-piece}H^0(\mathcal{L}_2^{-1} \otimes \mathcal{L}_1\otimes K_X) \xrightarrow{q^*}H^0(\Hom(E,\cL_1 \otimes K_X))
  \xrightarrow{i^*} H^0(K_X) \xrightarrow{\Lambda_3}H^1(\cL_2^{-1}\otimes \cL_1 \otimes K_X) \xrightarrow{}  \cdots \,.
\end{equation}
We can describe any element of $H^0(\Hom(E,\cL_1 \otimes K_X))$ as
\be\label{alphaelements}
\a = q^*(m) + \hat{v}\, , 
\ee
for some $m\in H^0(\mathcal{L}_2^{-1} \otimes \mathcal{L}_1\otimes K_X)$ and $\hat{v}\in H^0(\Hom(E,\cL_1 \otimes K_X))/q^*(H^0(\mathcal{L}_2^{-1} \otimes \mathcal{L}_1\otimes K_X))$, where the latter maps injectively into $H^0(K_X)$ under $i^*$.
This in turn gives rise to a Higgs field $\tilde{i}_*(\a)\in\End(E)\otimes K_X$, whose trace is:
	\be \label{traceformula}
	\tr \left(\tilde{i}_*(\a)\right) = \a \circ i = i^*(\hat{v}) \, ,
	\ee
where the first equality follows by direct local computation (see Lemma \ref{LemmaA1}), and in the last one we have used that $q \circ i = 0$.

We end our analysis of $H^0(\Hom(E,\cL_1 \otimes K_X))$ by discussing sufficient conditions for the existence of traceless Higgs fields with $m \neq 0$ or $\hat{v} \neq 0$. By injectivity of $\tilde{i}_*$ and $q^*$,  
\be \phi = \tilde{i}_*(q^*(m))
\ee
is a non-zero, nilpotent (because $q \circ i = 0$) Higgs field,
iff
$h^0(\cL_2^{-1}\otimes \cL_1 \otimes K_X) \neq 0$.
Moreover, we have that
\be \nonumber
\hat{v} \in H^0(\Hom(E,\cL_1 \otimes K_X))/q^*(H^0(\mathcal{L}_2^{-1} \otimes \mathcal{L}_1\otimes K_X))
\simeq i^*(H^0(\Hom(E,\cL_1 \otimes K_X))).
\ee
To see if there is a non-trivial element of this form we again consider the first row of the grid, given by \eqref{sections-first-piece}, but now the part involving the map $\Lambda_3$. We have that if
\be\label{suffcond1stuntw}
h^1(\cL_2^{-1}\otimes \cL_1 \otimes K_X) < h^0(K_X)\,,
\ee
then $\Lambda_3$ has a non-trivial kernel, or equivalently $ \dim(i^*(H^0(\Hom(E, \cL_1 \otimes K_X)))) \geq 1$. In this case there exist non-vanishing traceless Higgs fields of the form $\tilde{i}_*(\hat{v})$.

Finally, let us consider the third row of the grid in figure \ref{grid}. We have that any element of $\Ext^0(E,\cL_2 \otimes K_X)) \simeq H^0(\Hom(E,\cL_2 \otimes K_X))$ is of the form
\be
 q^*(w) + \hat{\psi}\,,
 \ee
with $w \in \Ext^0(\mathcal{L}_2,\mathcal{L}_2\otimes K_X)  \simeq H^0(K_X)$ and $\hat{\psi} \in H^0(\Hom(E,\cL_2 \otimes K_X))/q^*(H^0(K_X))$. 
Now, we have that $q^*(w) = \tilde{q}_*(\Id_E \otimes w)$ and so $\tilde{q}_*(H^0(K_X)) \simeq q^*(H^0(K_X))$. 
In other words, the Higgs fields that are not in the image of $\tilde{i}_*$ are of the type
\be\label{psiHiggs}
\Id_E \otimes w + \psi\,,
\ee
 where the traceless part $\psi$ gets injectively mapped by $\tilde{q}_*$ to
\be
\tilde{q}_*(\psi)=\hat\psi\in H^0(\Hom(E,\cL_2 \otimes K_X))/q^*(H^0(K_X))\simeq i^*(H^0(\Hom(E,\cL_2 \otimes K_X)))\, .
\ee
From the above, $\hat\psi$ is the counterimage under $i^*$ of an element $p\in H^0(\mathcal{L}_1^{-1} \otimes \mathcal{L}_2\otimes K_X)$. Elements of this form exist iff the matrix $\Lambda_+$ has a non-trivial kernel. Hence, similarly to before, a sufficient condition for non-zero traceless Higgs fields $\psi\in H^0(\End(E) \otimes K_X)/\tilde{i}_*(H^0(\Hom(E,\cL_1 \otimes K_X)))$ is that 
\be\label{suffcond2nduntw}
h^1(K_X)\,<\,h^0(\mathcal{L}_1^{-1} \otimes \mathcal{L}_2\otimes K_X)\,.
\ee

\bigskip
\noindent
{\em Summary.}
Putting all this together, we see that any traceless Higgs field $\phi \in H^0(\End_0(E) \otimes K_X)$
can be written as

\be 
\phi = \left( \tilde{i}_*(\a) - \frac{1}{2}\Id_E \otimes\, i^*(\a)  \right) + \psi \, ,
\label{finalphi}
\ee
for some $\a \in H^0(\Hom(E,\cL_1 \otimes K_X))$ and 
$\psi \in H^0(\End_0(E) \otimes K_X)$ such that
\be 
\tilde{q}_*(\psi) \in H^0(\Hom(E,\cL_2 \otimes K_X))/q^*(H^0(K_X)).
\ee
Note that $\tilde{i}_*(\a) \circ i: \cL_1 \rightarrow \cL_1 \otimes K_X \subset E \otimes K_X$. This means that:
\begin{itemize}
	\item $\cL_1$ is $\phi$-invariant whenever $\psi = 0$, and hence, in this case, the pair $(E,\phi)$ is Higgs stable iff $E$ is stable. Equivalently, only when $\psi \neq 0$ can we construct with this method a stable Higgs pair $(E,\phi)$ where the underlying bundle $E$ is unstable.
\end{itemize}
Moreover, $\tilde{i}_*(q^*(m)) \circ i = 0$ for all $m \in H^0(\cL_2^{-1}\otimes\cL_1 \otimes K_X)$. Hence:
\begin{itemize}
\item There exist nilpotent Higgs fields $\phi = \tilde{i}_*(q^*(m))  \in H^0(\End_0(E) \otimes K_X)$ for all $m \in H^0(\cL_2^{-1}\otimes\cL_1 \otimes K_X)$.
\end{itemize}
We also have that:
\begin{itemize}
	\item If $h^1(\cL_2^{-1}\otimes \cL_1 \otimes K_X) < h^0(K_X)$,
	there exist non-zero traceless Higgs fields of the form $\phi = \left( \tilde{i}_*(\hat{v}) -\tfrac12 \Id_E \otimes \, i^*(\hat{v})  \right)$, with $\hat{v}\in H^0(\Hom(E,\cL_1 \otimes K_X))/q^*(H^0(\mathcal{L}_2^{-1} \otimes \mathcal{L}_1\otimes K_X))$.
	\item If $h^1(K_X)\,<\,h^0(\mathcal{L}_1^{-1} \otimes \mathcal{L}_2\otimes K_X)$, there exist non-zero traceless Higgs field of the form $\phi=\psi$.
\end{itemize}

\subsubsection*{Twisting by an ideal sheaf}

The above discussion can be adapted to include an ideal sheaf twist of the extension \eqref{ext2}:
\begin{equation}\label{extis}
 \mathcal{L}_1\xrightarrow{i} E\xrightarrow{q} \mathcal{L}_2\otimes I_Z \,,
\end{equation}
where $I_Z$ is the ideal sheaf of a codimension-$2$ locus $Z \subset X$. For simplicity, and in view of the physics applications, let us work in the case where $X=S$ is a K\"ahler surface, so that $Z$ is a finite set of distinct points. As mentioned, this twist can be used to engineer vector bundles that cannot be continuously deformed to a split form, see Appendix \ref{TopNonSp}.

To begin with, note that now the last term of the sequence is a non-locally-free sheaf. Correspondingly $i$, thought of as a map between sheaves of sections, fails to be injective on $Z$. Similarly to the untwisted case \eqref{ext2}, the extension class defining $E$ is an element $\xi\in\Ext^1(\mathcal{L}_2\otimes I_Z,\mathcal{L}_1)$. This group however is more complicated, but can be described using the following long exact sequence
(see \cite[Chapter 2]{Friedman} for details):
\be
H^1(\mathcal{L}_2^{-1}\otimes\mathcal{L}_1)\xrightarrow{j} \Ext^1(\mathcal{L}_2\otimes I_Z,\mathcal{L}_1)\xrightarrow{\pi}\mathbb{C}^{|Z|}\xrightarrow{Y} H^0(\mathcal{L}_1^{-1}\otimes\mathcal{L}_2\otimes K_S)\xrightarrow{} \cdots
\label{seqY}
\ee
where we have used that $Ext^1(\mathcal{L}_2\otimes I_Z,\mathcal{L}_1)\simeq \mathcal{O}_Z$ is the structure sheaf on $Z$, and thus $H^0(Ext^1(\mathcal{L}_2\otimes I_Z,\mathcal{L}_1))\simeq \mathbb{C}^{|Z|}$. Moreover, $H^2(\mathcal{L}_2^{-1}\otimes\mathcal{L}_1)\simeq H^0(\mathcal{L}_1^{-1}\otimes\mathcal{L}_2\otimes K_S)$ by Serre duality. 
Therefore the extension class can be written as
\be\label{extclassideal}
\xi = j(a_+) + \hat{\xi}\,,
\ee
where $a_+\in H^1(\mathcal{L}_2^{-1}\otimes\mathcal{L}_1)$, and $\hat{\xi}\in\Ext^1(\mathcal{L}_2\otimes I_Z,\mathcal{L}_1)/j(H^1(\mathcal{L}_2^{-1}\otimes\mathcal{L}_1))$, which has the property of having a non-zero value on $Z$, i.e.~$\pi(\hat{\xi})\neq0$. Note that the existence of $\hat{\xi}$ is obstructed if the map $Y$ is injective. 

The strategy to determine the components of a traceless Higgs field compatible with \eqref{extis} is similar to the one applied to \eqref{ext2}, so in the following we will simply point out the main differences. The first row of the grid in Figure \ref{grid} is replaced by:
\be\label{1strowideal}
H^0(\mathcal{L}_2^{-1}\otimes\mathcal{L}_1\otimes K_S)\xrightarrow{q^*} \Ext^0(E,\mathcal{L}_1\otimes K_S)\xrightarrow{i^*}H^0(K_S\otimes I_Z)\xrightarrow{\Lambda_3} H^1(\cL_2^{-1}\otimes \cL_1 \otimes K_S)\xrightarrow{} \cdots
\ee
see Appendix \ref{ap:homalgebra} for details.

Considering the elements $\alpha\in H^0(\Hom(E,\cL_1 \otimes K_S))$ as in \eqref{alphaelements}, all Higgs fields of the form $ \tilde{i}_*(\a)$ vanish on $Z$ because of the map $i$. Moreover, while nothing changes regarding the existence of those of the form $ \tilde{i}_*(q^*(m))$ compared to the untwisted case, some differences arise for Higgs fields of the form $ \tilde{i}_*(\hat{v})$. Indeed the elements $\hat{v}$ now map injectively under $i^*$ into $H^{0}(K_S\otimes I_Z)$, which are sections of the canonical bundle that vanish on $Z$. Moreover, not all such sections will lift to a Higgs field, but only those in the kernel of the map $\Lambda_3$. Similarly to the untwisted case (cf. \eqref{suffcond1stuntw}), a sufficient condition for them to exist is
\be 
h^1(\cL_2^{-1}\otimes \cL_1 \otimes K_S) < h^0(K_S\otimes I_Z)\,,
\label{suffvis}
\ee
which guarantees that the kernel of $\Lambda_3$ in \eqref{1strowideal} (and thus the image of $i^*$) is non-trivial.

Finally, the third row of the grid in Figure \ref{grid} becomes:
\be\label{3strowideal}
H^0(K_S)\xrightarrow{q^*} \Ext^0(E,\mathcal{L}_2\otimes K_S\otimes I_Z)\xrightarrow{i^*}H^0(\mathcal{L}_1^{-1}\otimes \mathcal{L}_2 \otimes K_S\otimes I^2_Z)\xrightarrow{\Lambda_+} H^1(K_S)\xrightarrow{} \cdots
\ee
(see Appendix \ref{ap:homalgebra} for details).
This tells us that, traceless Higgs fields of the type $\psi$ as in \eqref{psiHiggs} must be such that $p=i^*(\tilde{q}_*(\psi))\in\ker\Lambda_+\subset H^0(\mathcal{L}_1^{-1}\otimes\mathcal{L}_2\otimes K_S\otimes I^2_Z)$. In particular, $p$ is a section of $\mathcal{L}_1^{-1}\otimes\mathcal{L}_2\otimes K_S$ that vanishes quadratically on $Z$. However, not all such $p$ will lift to a Higgs field, but only those in the kernel of $\Lambda_+$.  
Hence, similarly to Eq. \eqref{suffcond2nduntw} for the untwisted case, a sufficient condition for the existence of Higgs fields of the type $\psi$ is
\be
h^1(K_S)\,<\,h^0(\mathcal{L}_1^{-1} \otimes \mathcal{L}_2\otimes K_S\otimes I_Z^2)\,.
\label{suffpis}
\ee

\subsubsection*{Physical interpretation}

When applied to Vafa-Witten systems, for which $X$ is a K\"ahler surface $S$, the results of this section have a nice interpretation in terms of the physics of 7-branes.   Higgs-bundle polystability guarantees that, given a pair $(E,\phi)$ on $S$, a solution will be found for \eqref{Dterm}, known as 7-brane D-term equation in the physics literature. However, from the physics viewpoint, which usually employs a unitary gauge, it still remains to show that the other two equations in \eqref{VW} allow for a pair $(E,\phi)$ in which $\phi$ is non-trivial. These two equations \eqref{Fterm1} and \eqref{Fterm2} are known as 7-brane F-term equations, because they can be obtained from extremisation of the following superpotential \cite{Jockers:2005zy,Martucci:2006ij,Donagi:2008ca,Beasley:2008dc}
\be
W\, =\, \int_S \tr \left(\mathbb{F} \wedge \phi\right)\, .
\label{supo}
\ee
%

%Indeed, if we assume that $X$ is a K\"ahler surface $S$ then Theorem \ref{VT} gives the necessary condition ${\rm deg}\, K_S \geq 0$ to solve the Vafa-Witten equations \eqref{VW} with $\phi \neq 0$, while Corollary \ref{VTC} requires ${\rm deg}\, K_S > 0$ on solutions where $[\phi,\phi^\dagger] \neq 0$. 

From the mathematics perspective, a non-trivial $\phi$ amounts to the existence of certain holomorphic sections, whose precise definition depends on the choice of bundle $E$. At the beginning of this section we have described the holomorphic sections that define a Higgs field $\phi$ compatible with a bundle extension of the form \eqref{ext2}. In doing so we have illustrated how the different components of $\phi$ are determined by the topological data of $E$ and $S$. In the following we would like to show that our results have a direct interpretation in terms of the restrictions that the superpotential \eqref{supo} imposes on $\phi$.

To make the connection it is useful to first consider the case where the extension \eqref{ext2} is trivial, in the sense that the bundle $E$ is of the split type \eqref{split}. Then it is straightforward to describe the holomorphic traceless deformations of the bundle connection $\mathbb{A}$ as
\be
\delta\mathbb{A}^{(0,1)}=\left(\begin{array}{cc} a_3 &a_+\\ a_-& -a_3 \end{array}\right)\, ,
\label{fluctua}
\ee
where $a_3 \in H^1(\cO_S)$, $a_+ \in \Ext^1 (\mathcal{L}_2,\mathcal{L}_1)$ and $a_- \in \Ext^1 (\mathcal{L}_1,\mathcal{L}_2)$. Let us for now assume that $S$ is simply connected, and so there are no holomorphic deformations along the diagonal. For $\deg(\mathcal{L}_1) = \deg(\mathcal{L}_2)$ the bundle is polystable by itself, and the Vafa-Witten equations admit a solution with vanishing Higgs field. The traceless deformations of $\phi$ are given by
\be
\delta\phi=  \left(\begin{array}{cc} {v}& {m}\\ {p}&- {v}\end{array}\right)\,,
\label{fluctuphi}
\ee
where 
\be
m\in \Ext^0(\mathcal{L}_2,\mathcal{L}_1\otimes K_S)\, , \qquad p \in \Ext^0(\mathcal{L}_1,\mathcal{L}_2\otimes K_S)\, , \qquad v \in H^0(K_S)\, .
\label{phidef}
\ee
If instead $\deg(\mathcal{L}_1) < \deg(\mathcal{L}_2)$ the bundle is not stable in the Hermite-Yang-Mills sense, and the deformation $m$ has to be turned on to guarantee Higgs stability, leading to the nilpotent solution of the form \eqref{NilpPhi}. Bundle and Higgs deformations around this nilpotent solution are also of the form \eqref{fluctua} and \eqref{fluctuphi}. As one can check, the deformations \eqref{phidef} correspond to specific elements of the grid in Figure \ref{grid}. 

Clearly, the fact that all these deformations are holomorphic sections is a consequence of the F-term equations \eqref{Fterm1} and \eqref{Fterm2}, but it is not the only one. F-terms also impose constraints on performing more than one deformation simultaneously. This can be easily encoded in the trilinear terms involving fluctuations that arise from the superpotential \eqref{supo}, dubbed Yukawa couplings. The ones relevant for the present setup have been computed in \cite[Section 3]{Marchesano:2019azf}, and arise from the following piece of the superpotential
\be
W \supset  \int_S a_+ \wedge a_- \wedge v  =  \tilde \Lambda_3^{in\a} a_{+\, i} a_{-\, n} v_{\a}\, ,
\label{yuk}
\ee
where we have expanded each of the sections $a_+$, $a_-$, and $v$ in a basis of the corresponding cohomology groups, as $a_+ = a_{+\, i} \psi_+^i$, $a_- = a_{-\, n} \psi_-^n$ and $v = v_\a \chi_3^\a$ with $a_{+\, i}, a_{-\, n}, v_\a \in \IC $. In physics terms, these complex numbers represent the vacuum expectation values (vevs) of the four-dimensional fields, and therefore the local field space. Finally we have defined
\be
\tilde \Lambda^{in\a}_{3}=\int_S \psi_+^i\wedge\psi_-^n\wedge\chi_{3}^{\a}\,.
\ee
The F-flatness constraints derived from the Yukawa couplings \eqref{yuk} are the following \cite{Marchesano:2019azf} 
\be
 \tilde \Lambda^{i n \alpha}_{3} a_{+\, i} a_{-\, n} = \tilde \Lambda^{in \a}_{3} a_{-\, n} v_{\a}= \tilde \Lambda^{i n \a}_{3} a_{+\, i} v_{\a} =0\,,
\label{condyuk}
\ee
and so, e.g., care must be taken when performing a simultaneous deformation along the directions $a_+$ and $v$. Other deformations like $m$ and $p$ in \eqref{phidef} are on the other hand unobstructed by the above Yukawa couplings. 

Now, the extension \eqref{ext2} is described by an element of $H^1(\mathcal{L}_2^{-1} \otimes \mathcal{L}_1) \simeq \Ext^1 (\mathcal{L}_2,\mathcal{L}_1)$, and in this sense it can be understood as a deformation of the split bundle \eqref{split} along the direction $a_+$. Let us denote by $\langle a_+\rangle = \langle a_{+\, i}\rangle \psi_+^i \in \Ext^1 (\mathcal{L}_2,\mathcal{L}_1)$ the element describing a particular bundle extension. Then  $\Lambda_3^{n \a} \equiv  \tilde \Lambda^{i n \a}_{3} \langle a_{+\, i}\rangle$ defines a linear map
\be
\Lambda_3 \, :\, H^0(K_S) \raw  \Ext^1(\mathcal{L}_2,\mathcal{L}_1\otimes K_S)\simeq H^1 (\mathcal{L}_1^{-1}\otimes\mathcal{L}_2)\, ,
\label{lambda3}
\ee
where we have used Serre duality. The last equality in \eqref{condyuk} then translates into the statement that the Higgs field deformations along the direction $v$ allowed by the F-terms must belong to $ \ker \Lambda_3$. To sum up, given a vector bundle $E$ defined by the extension \eqref{ext2} on a simply-connected surface $S$, the Higgs field deformations allowed by Yukawa couplings are those of the form \eqref{fluctuphi} with arbitrary $m$, $p$, and with $v \in \ker \Lambda_3$. 

This result nicely matches the characterisation of traceless Higgs fields compatible with the extension \eqref{ext2}, as summarised by \eqref{finalphi} and the statements below. There are essentially three contributions to \eqref{finalphi}, that must now be identified with different elements of the grid in Figure \ref{grid}, and that can be put in correspondence with sections of the form \eqref{phidef}. The contribution $\phi = \tilde{i}_*(q^*(m))$ is in one-to-one correspondence with elements of $\Ext^0(\mathcal{L}_2,\mathcal{L}_1\otimes K_S)$ and gives rise to a nilpotent Higgs field, as in the split case. The contribution $\phi = \left( \tilde{i}_*(\hat{v}) -\tfrac12 \Id_E \otimes \,i^*(\hat{v})  \right)$ corresponds to elements of $H^0(K_S)$ that belong to the kernel of the map $\Lambda_3$ in the grid. This precisely matches our F-term analysis if we identify such a map with \eqref{lambda3}. Finally, the contribution $\phi = \psi$ corresponds to sections of the form $p \in \Ext^0(\mathcal{L}_1,\mathcal{L}_2\otimes K_S)$ which are in the kernel of the map $\Lambda_+$ in the grid. As we have assumed that $H^1(\cO_S)=0$, this map is such that the kernel is the whole of $\Ext^0(\mathcal{L}_1,\mathcal{L}_2\otimes K_S)$, and so these deformations are unconstrained. If on the other hand we consider a surface such that $H^1(\cO_S)\neq 0$, then new Yukawa couplings will be developed from the superpotential \eqref{supo} leading to an obstruction analogous to the one found above. 

Indeed, the Yukawa couplings for this more general case has been worked out in \cite[Appendix C]{Marchesano:2019azf}, and result in
\be
W\, \supset \,  \int_S a_+ \wedge a_- \wedge v  +   \int_S a_3 \wedge a_+ \wedge p -   \int_S a_3 \wedge a_- \wedge m \, .
%\\ \nonumber & = &   \tilde \Lambda_3^{ij\a} a_{+\, i} a_{-\, j} v_{\a} +   \tilde \Lambda_+^{ij\a} a_{3\, i} a_{+\, j} p_{\a} -   \tilde \Lambda_-^{ij\a} a_{3\, i} a_{-\, j} m_{\a} \, ,
\label{yuk2}
\ee
Similarly to the map $\Lambda_3$, one can construct a map $\Lambda_+$ from the second term in \eqref{yuk2}, and the element $\langle a_+\rangle \in \Ext^1 (\mathcal{L}_2,\mathcal{L}_1)$ describing the bundle extension:
\be
\Lambda_+ \, :\, \Ext^0(\mathcal{L}_1,\mathcal{L}_2\otimes K_S) \raw H^1 (\cO_S) \simeq H^1 (K_S)\, .
\label{lambda+}
\ee
The Higgs-field deformations allowed by the F-term analysis are those that belong to the kernel of this map. This precisely matches the result for the contribution of the form $\phi = \psi$ in our previous analysis, upon identifying \eqref{lambda+} with the map $\Lambda_+$ in the grid. Finally, as the third term in \eqref{yuk2} does not contain $a_+$, no further constraint arises from it, and so as expected the Higgs-field deformation coming from $m$ remains unconstrained. 

One of the advantages of this F-term analysis is that it is easily generalised to other setups, like for instance bundles of higher rank on a K\"ahler surface $S$. Indeed, if one is able to understand the bundle $E$ as a deformation of a split bundle, then one may identify the unobstructed Higgs field deformations in terms of the Yukawa couplings of the configuration, and therefore construct the most general Higgs field. The only drawback of this approach is that it is only reliable for small deformations of the pair $(E, \phi)$ around the split solution. Nevertheless, the F-term analysis could give valuable insight on how to generalise the mathematical construction for the above rank-$2$ bundle to other setups. \\

The ideal-sheaf twisting \eqref{extis} also has a nice interpretation in terms of D-brane physics. This time, on top of the 7-branes wrapping $S$ one adds $|Z|$ D3-branes (with $|Z|$ the number of points in $Z$), and switches on the vevs of the fields in the 37-sector to form a bound state.\footnote{See \cite{Collinucci:2014qfa} for a different analysis of the same system.} The 37-sector is described by two chiral multiplets at each D3-brane location, with opposite charge under the U(1) of the D3-brane, and transforming as a doublet under the 7-brane rank-$2$ bundle. Let us label them as
\be
\Phi^K_{73} = 
 \left(
 \begin{matrix}
\phi_{73\, -}^K \\
\phi_{73\, +}^K
\end{matrix}
\right)\, ,
\qquad
\Phi^K_{37} = 
 \left(
 \begin{matrix}
\phi_{37\, +}^K \\
\phi_{37\, -}^K
\end{matrix}
\right)\, ,
\ee
with $K$ running over the $|Z|$ points ${Q}_K$ where the D3-branes are located, which we assume near the disjoint set $Z = \{P_K\}_{K=1,\ldots,|Z|}$. The $\pm$ above indicates charge $\pm1$ under the relative U(1) of the 7-brane stack, assigning positive charge to the first 7-brane and negative charge to the second 7-brane. The superpotential including these modes reads 
\bea
\label{yuk37}
W & \supset & \int_S a_+ \wedge a_- \wedge v  +   \int_S a_3 \wedge a_+ \wedge p -   \int_S a_3 \wedge a_- \wedge m \\ \nonumber 
& + &  \sum_K\left[  v(Q_K)  \left(\phi_{73\, -}^K  \phi_{37\, +}^K   - \phi_{73\, +}^K \phi_{37\, -}^K\right) +   m(Q_K)  \phi_{73\, -}^K  \phi_{37\, -}^K   +   p(Q_K) \phi_{73\, +}^K \phi_{37\, +}^K\right] \, ,
\eea
where $p(Q_K)$ stands for the value of the section $p$ at a point $Q_K$, same for $v(Q_K)$ and $m(Q_K)$. The second line does not involve any integral, because the 37-sector fields are $\delta$-function localised at the corresponding D3-brane location. In the brane system this location is not fixed, and depends on neutral complex fields $\phi_{33, K}$ of each D3-brane. Therefore, one should understand terms like $p(Q_K)$ as a Taylor expansion
\be
p(Q_K) \equiv \left[e^{\phi_{33, K} \cdot \p}\, p \right] (P_K) =  p (P_K) +  \p p (P_K)\, \phi_{33, K} + \dots\,, 
\label{taylor}
\ee
where $\phi_{33, K} = Q_K - P_K$ measures the separation of the $K^{\rm th}$ D3-brane from $P_K$. Additionally the section $p$ can be decomposed on a basis $\{\chi_-^\lambda\} \in H^0(\mathcal{L}_1^{-1}\otimes\mathcal{L}_2\otimes K_X)$ as $p = p_\lambda \chi^\lambda_-$, with $p_\lambda \in \IC$ representing the four-dimensional-field vev, so that $ \p_i p (P_K) = p_\lambda \p \chi^\lambda_- (P_K)$, etc. Similar expansions hold for $v(Q_K)$ and $m(Q_K)$, from where it follows that the second line of \eqref{yuk37} not only contains cubic couplings on the fields $v_\a, m_\sigma, p_\lambda, \phi_{37\pm}^K, \phi_{73\pm}^K$,  but also higher order couplings involving the fields $\phi_{33, K}$.

Taking this into account, one can understand the constraints found for the twisted bundle \eqref{extis} as switching on vevs to the new fields in the 37-sector, and then imposing the F-flatness constraints coming from \eqref{yuk37}. More precisely, to match our previous results we must switch on equal non-vanishing vevs for the fields $\phi_{73\, +}^K$ and $ \phi_{37\, +}^K$ $\forall K$, represented by $ \phi_{+}^K \in \IC$, while keeping  the vevs of $\phi_{73\, -}^K,  \phi_{37\, -}^K, \phi_{33, K}$ to zero.
%$\langle \phi_{+}^K \rangle \equiv  \langle \phi_{73\, +}^K \rangle = \langle \phi_{37\, +}^K \rangle \in \IC$ respectively, while keeping  $\langle \phi_{73\, -}^K \rangle = \langle \phi_{37\, -}^K \rangle = 0$.

Under this assumption, one finds that the F-flatness conditions for the fields $p_\lambda$ read
\be
 \sum_K \chi^\lambda_- (P_K) \,   ( \phi_{+}^K )^2  = 0\, , \quad \forall \lambda\, .
\ee
This imposes $h^0(\mathcal{L}_1^{-1}\otimes\mathcal{L}_2\otimes K_X)$ constraints on the values for $\phi_{+}^K$, so generically we will not be able to have $\phi_{+}^K \neq 0$ unless $|Z| > h^0(\mathcal{L}_1^{-1}\otimes\mathcal{L}_2\otimes K_X)$. This can be understood as the physical counterpart of the constraint imposed on $\hat{\xi}$ by the map $Y$ on the sequence \eqref{seqY}, by identifying $\{\phi_+^K\}_{K=1,\ldots,|Z|}=\pi(\hat{\xi})$.

Regarding the constraints found for constructing the Higgs field, they can also be understood in terms of \eqref{yuk37}. On the one hand, the F-flatness conditions corresponding to the fields $\phi_{73\, -}^K$ and $\phi_{37\, -}^K$ read
\be
v (P_K)\,  \phi_{+}^K = 0 \, , \quad \forall K\, .
%v_\alpha \chi_3^\a|_{q_K} \langle \phi_{73\, +}^K \rangle  = v_\alpha \chi_3^\a|_{q_K} \langle \phi_{37\, +}^K \rangle = 0 \, .
\label{vyukis}
\ee
Because by assumption $\phi_{+}^K \neq 0$, this implies that $v (P_K) =0$, which reproduces that Higgs fields are constructed from sections of the canonical that vanish on $Z$, $i^*(\hat{v})\in H^0(K_S\otimes I_Z)$. T hese sections must be such that the F-terms for $a_-$ arising from the first line of \eqref{yuk37} also vanish, which leads to $i^*(\hat{v}) \in \ker \Lambda_3$ and to the sufficient condition \eqref{suffvis}. Notice that the fields $\phi_{73\, -}^K$, $\phi_{37\, -}^K$ also couple to the sections $m$, but since they do it quadratically and their vev vanish no constraint is imposed on them, as expected from our  analysis. 

Conversely, the sections $p \in H^0(\mathcal{L}_1^{-1}\otimes\mathcal{L}_2\otimes K_X)$ get the most stringent constraint, as they couple quadratically to the fields $\phi_{73\, +}^K$, $\phi_{37\, +}^K$. The F-terms for the latter vanish if
\be
p (P_K)\,   \phi_{+}^K   = 0\, , \quad \forall K\, ,
\ee
which selects sections that vanish on $Z$. While this is similar to the constraint \eqref{vyukis}, now an extra one appears due to the quartic couplings coming from the expansion \eqref{taylor}. One finds that the F-flatness condition for the fields $\phi_{33, K}$ reads 
\be
\p p (P_K)\,   ( \phi_{+}^K )^2 = 0\, , \quad \forall K\, ,
\ee
or in other words the sections $p$ vanish quadratically on $Z$. We therefore reproduce our previous result that $p \in H^0(\mathcal{L}_1^{-1}\otimes \mathcal{L}_2 \otimes K_X\otimes I^2_Z)$. Finally, due to the first line in \eqref{yuk37} we have that $p \in \ker \Lambda_+$, from where the sufficient condition \eqref{suffpis} is obtained.

\section{Explicit examples}
\label{s:examples}

If $S$ is a K\"ahler surface, by Corollary \ref{VTC} the existence of T-branes forces it to be properly elliptic or of general type. In this section we consider the case of a very simple properly elliptic surface that is a product of two curves. The construction presented below can nonetheless be easily adapted to more general properly elliptic surfaces.

Suppose that $S = C \times T$ where $C$ is a curve of genus $\geq 2$ and $T$ is an elliptic curve. Then, $S$ is a properly elliptic surface with trivial elliptic fibration given by projection onto the first factor $\pi := {\rm pr}_1: S \rightarrow C$. Note that $K_S \simeq \pi^*K_C$. We assume for now that $C$ is {\em not} hyperelliptic, implying that $g \geq 3$. 

For any $P_0 \in C$, we set $T_0 := \{ P_0 \} \times T$ and 
\begin{equation} 
K_S(nT_0) := \mathcal{O}_S(nT_0) \otimes K_S
\end{equation}
for all $n \in \mathbb{Z}$. 
Consider a rank-2 vector bundle $E$ on $S$ that fits into the exact sequence
\begin{equation}\label{example}
0 \rightarrow \cL \xrightarrow{i} E \xrightarrow{q} \cL^{-1} \otimes I_Z \rightarrow 0\, ,
\end{equation}
with $\cL = \mathcal{O}_S(-sT_0)$ for some positive integer $s$ and $Z$ a finite set of distinct points on $S$. 
Note that $c_1(E) = 0$ and $c_2(E) = l(Z)$,
where $l(Z)$ is the number of points in $Z$.
Moreover, $E$ is stable. This is because the only line bundles mapping into $E$ are of the form $\cL(-D)$ for some effective divisor $D$ on $S$ and 
\begin{equation}
\mu(\cL(-D)) = \deg \cL(-D) < 0 = \deg(E)/2 = \mu(E)\,. 
\end{equation}
Any Higgs field $\phi \in H^0(S,\End(E) \otimes K_S)$ of $E$ thus gives rise to a stable pair $(E,\phi)$. We are interested in traceless ones, which correspond to the elements of $H^0(\End_0(E) \otimes K_S)$.

Referring to section \ref{s:sufficient}, any element $\phi \in H^0(\End_0(E) \otimes K_S)$ can be written as
\begin{equation} 
\phi = (\tilde{i}_*(\alpha) - \tfrac12 \Id_E \otimes i^*(\alpha)) + \psi
\end{equation}
for some $\alpha \in H^0(\Hom(E, \cL \otimes K_S))$ and $\psi \in H^0(\End_0(E) \otimes K_S)$ such that 
\begin{equation} 
\tilde{q}_*(\psi) \in H^0(\Hom(E,\cL^{-1} \otimes I_Z \otimes K_S))/q^*(H^0(K_S))\,,
\end{equation}
where $\tilde{i} = i \otimes \Id_{K_S}$ and $\tilde{q} = q \otimes \Id_{K_S}$. 
Moreover, $\alpha$ can be written as 
\begin{equation}  
\alpha  = q^*(m) + \hat{v}  
\end{equation}
for some
\begin{equation}
m \in H^0(\cL^2 \otimes K_S)\,
\end{equation} 
and
\begin{equation} 
\hat{v} \in H^0(\Hom(E, \cL \otimes K_S))/q^*(H^0(\cL^2 \otimes K_S))\,. 
\end{equation}

In this case,
\begin{equation}  
H^0(\cL^2 \otimes K_S) = H^0(S,K_S(-2sT_0))
\end{equation}
and
\begin{equation}
h^0(S,K_S(-2sT_0)) \geq g-2s
\end{equation}
by Proposition \ref{cohomology-on-surface}.
Therefore, $H^0(\cL^2 \otimes K_S) \neq 0$ when $s < 2g$,
in which case there exist traceless Higgs fields with $m \neq 0$. In particular, such Higgs fields exist when $s=1$ since $g \geq 3$. Moreover, $[\phi,\phi^\dagger] \neq 0$ whenever $m \neq 0$.

On the other hand,
\begin{equation} 
h^0(S ,K_S \otimes I_Z) < g
\end{equation}
since $h^0(S,K_S) = h^0(C,K_C) =g$ and the canonical linear system $|K_C|$ has no base points because $g \geq 2$ (see \cite[IV, Lemma 5.1]{Hartshorne}). Moreover,
\begin{equation} 
h^1(S,\cL^2 \otimes K_S) = h^1(S,K_S(-2sT_0)) \geq g-2s + 1 \geq g-1
\end{equation}
by Proposition \ref{cohomology-on-surface}.
Consequently, $h^1(S,\cL^2 \otimes K_S) \geq  h^0(S ,K_S \otimes I_Z)$ for all $s \geq 1$, implying that there is no guarantee that there exist traceless Higgs fields with $\hat{v} \neq 0$ in that case.

To summarise, we have:

\begin{proposition}
	The rank-2 vector bundle $E$ given by the non-trivial extension \eqref{example} is stable with $c_1(E) = 0$ and $c_2(E) = l(Z)$. In particular, $(E,\phi)$ is a stable Higgs pair for all Higgs fields $\phi \in \Gamma(\End(E) \otimes K_S)$. Moreover, $E$ admits non-zero traceless Higgs fields $\phi$ such that $[\phi,\phi^\dagger] \neq 0$ whenever $s < g/2$. 
\end{proposition}

\begin{remark}
	Note that the above construction also works when the curve $C$ is hyperelliptic and $s=1$. The computations are similar and we only state the main points.
	
	First of all, because the curve is hyperelliptic, its genus can be 2. Therefore, $g \geq 2$. Moreover, it admits a double covering of $\mathbb{P}^1$, implying that there exists a meromorphic function on $C$ of degree 2 (although it still does not admits a meromorphic function of degree 1 since it is not isomorphic to $\mathbb{P}^1$). In this case, we thus have
	\begin{equation} 
	h^0(C,\mathcal{O}_C(2P_0))) = 2
	\end{equation}
	for all $P_0 \in C$. By Riemann-Roch, $h^1(C,\mathcal{O}_C(2P_0)) = g-1$ so that 
	\begin{equation} 
	h^0(S,K_S(-2T_0)) = g-1 > 0
	\end{equation}
	and 
	\begin{equation} 
	h^1(S,K_S(-2T_0)) = g+1 > g =  h^0(S,K_S) > h^0(S,K_S \otimes I_Z).
	\end{equation}
	This means that the stable rank-2 vector bundle given by \eqref{example} with $s=1$ admits a non-zero traceless Higgs field $\phi$ with $m \neq 0$ so that $[\phi,\phi^\dagger] \neq 0$. There is, however, again no guarantee that there exists a traceless Higgs field such that $\hat{v} \neq 0$.  
\end{remark}

\bigskip
\noindent
{\em Summary of examples.} Suppose that $S = C \times T$ with $C$ a curve of genus $g \geq 2$ and $T$ an elliptic curve. Let us assume that $\det E \simeq \mathcal{O}_S$ so that $c_1(E)=0$. We have two cases:
\begin{enumerate}
	\item {\em $C$ is hyperelliptic:}
	\begin{itemize}
		\item If $g=2$ and $c_2 = 0$, then one can pullback Hitchin's examples of stable Higgs pairs on $C$ to $S$, see \cite[Example 3.13]{Hitchin:1986vp}. 
		\item If $g \geq 2$ and $c_2 \geq 1$, consider the rank-2 vector bundle $E$ given by
		\begin{equation} 
		0 \rightarrow \mathcal{O}_S(-T_0) \rightarrow E \rightarrow \mathcal{O}_S(T_0) \otimes I_Z \rightarrow 0
		\end{equation}
		with $Z$ a finite set of $c_2$ distinct points in the support of $K_S$. Then, $E$ is stable and admits a non-zero traceless Higgs field $\phi$ with $[\phi,\phi^\dagger] \neq 0$.
	\end{itemize}
	\item {\em $C$ is not hyperelliptic:}  In this case, $g \geq 3$ (because every curve of genus 2 is hyperelliptic). Then, if $c_2 \geq 1$, the rank-2 vector bundle $E$ given by
	\begin{equation} 
	0 \rightarrow \mathcal{O}_S(-sT_0) \rightarrow E \rightarrow \mathcal{O}_S(sT_0) \otimes I_Z \rightarrow 0
	\end{equation}
	with $Z$ a finite set of $c_2$ distinct points in the support of $K_S$ and $s$ a positive integer. Then, $E$ is stable and admits a non-zero traceless Higgs field $\phi$ with $[\phi,\phi^\dagger] \neq 0$ whenever $s < g/2$. %It also admits a traceless Higgs field with non-zero $\hat{v}$ when $s=1$, and non-zero $\psi$ when $1<s<g/2$.
	
\end{enumerate}

\section{Conclusions}
\label{s:conclu}

In this paper we have proven a new vanishing theorem for regular solutions to the Vafa-Witten equations \cite{Vafa:1994tf}: No solutions exist with $[\phi,\phi^\dagger]\neq0$ on compact K\"ahler surfaces with either positive or vanishing Ricci curvature. Our main result is expressed in Theorem \ref{VT}, which in fact holds more generally for complex manifolds of any dimensions. Contrary to previous findings in the math and the physics literature \cite{Tanaka:2017jom,Tanaka:2017bcw,Marchesano:2017kke}, we have made no assumptions on the Higgs pair $(E,\phi)$, such as $\mathbb{C}^*$-fixed Higgs-field configuration or split vector bundle. Our conclusions simply follow from holomorphicity of the Higgs field and polystability of the Higgs pair, as required by the Hitchin-Kobayashi correspondence.

One should stress that our vanishing theorem is different in nature to the one derived in \cite{Vafa:1994tf}. There, in the case of a K\"ahler surface $S$, it was found that a necessary condition for solutions with $\phi \neq 0$ is that $\int_S R \tr (\phi_u\phi_u^\dag) \leq 0$, with $R$ the scalar curvature and $\phi_u$ the Higgs field in the unitary gauge. In cases where the presence of $\tr (\phi_u\phi_u^\dag)$ does not change the sign of the integral, like for instance when $R$ does not change sign along $S$, this matches with the first part of Theorem \ref{VT} applied to K\"ahler surfaces, because  $\int_S R\, \d{\rm vol} = - \deg(K_S)$. In general, however, these are two different conditions, whose precise connection would be interesting to understand. For surfaces such that $\int_S R \, \d{\rm vol} >0$ our vanishing result implies that there is no Higgs field solving the Vafa-Witten equations, and so necessarily $\int_S R \tr (\phi_u\phi_u^\dag) = 0$. When $K_S$ is trivial our results imply that $\phi_u$ is proportional to the identity, and so again $\int_S R \tr (\phi_u \phi_u^\dag)  = 0$. Finally, when $\int_S R \, \d{\rm vol} < 0$ our analysis does not provide any obvious statement on $\int_S R \tr (\phi_u \phi_u^\dag)$, although it leads to the inequalities \eqref{ineq}.  Therefore, it would be particularly interesting to develop the interplay between both vanishing results in this case.

Our results have important implications for the physics of wrapped D-branes in string theory. More precisely, solutions with $[\phi,\phi^\dagger]\neq0$ correspond to special supersymmetric brane bound states, called T-branes \cite{Cecotti:2010bp}, whereby $\phi$ is identified with the field describing deformations of a stack of D-branes in two real transverse dimensions. The conditions on the Ricci curvature we found are interpreted as geometric constraints imposed by the mere existence of such T-branes as stable BPS vacua in string compactifications.
T-branes are nowadays a key ingredient in string theoretic constructions of phenomenologically viable models of Grand Unification. In such constructions, T-branes may be required to wrap manifolds with positive Ricci curvature. Our results forbid this circumstance, unless there are regions where the Higgs field fails to be holomorphic. Indeed, poles may be induced by considering defects, which in turn correspond to additional D-branes present in the compactification and intersecting the T-brane. This situation has been explored in \cite{Marchesano:2019azf} in the special case of vector bundles topologically given by sum of line bundles. It would be important to extend these setups to general meromorphic Higgs pairs, using the techniques developed in the present paper.

In Section \ref{s:sufficient} we have given a complete characterisation of rank-$2$ stable Higgs pairs on a projective variety, dividing the discussion in two cases: The extension by a line bundle of $a)$~another line bundle and $b)$~of a non-locally-free sheaf. The first case constitutes a class of topologically split vector bundles but holomorphically non-split, whereas the second class contains topologically non-split bundles, due to the presence of the ideal sheaf. In both cases, we describe the Higgs field using a grid of long exact sequences in cohomology, deriving sufficient topological conditions for its existence. We finally provided a physical interpretation of these cohomological results, in terms of superpotentials and holomorphic Yukawa couplings, which characterise the low-energy effective physics of  D7-branes wrapped on surfaces. 

Particularly intriguing from the physics perspective is the ideal sheaf twist, which we interpreted as the result of having coupled the D7-brane system to a system of D3-branes located at points on the K\"ahler surface. These sorts of D3/D7 bound states have been argued in \cite{Collinucci:2014qfa} to be the origin of the so-called point-like matter,\cite{Cecotti:2010bp}, one of the most exotic features of T-branes. While our analysis appears to be compatible with the findings of \cite{Collinucci:2014qfa}, it would be interesting to investigate the connection to point-like matter further, especially by including the information of D-terms. Moreover, the physical meaning of the extension class \eqref{extclassideal} remains rather obscure, in particular for what concerns the role of $\hat{\xi}$, carrying the information of the localised gluing modes, and how the latter appears in the gauge connection. We hope to come back to these issues in the near future.

\bigskip

\centerline{\bf  Acknowledgments}

\bigskip

We would like to thank Luis \'Alvarez-C\'onsul, Lara Anderson, Andr\'es Collinucci, Mario Garc\'ia-Fern\'andez, Oscar Garc\'ia Prada, Yuuji Tanaka, Richard P. Thomas and Eric Sharpe for useful discussions. FM is supported by the Spanish Research Agency through the grants SEV-2016-0597 and PGC2018-095976-B-C21 from MCIU/AEI/FEDER, UE. RM was partially supported by an NSERC Discovery Grant. RS is supported by the program Rita Levi Montalcini for young researchers (D.M. n. 975, 29/12/2014). We gratefully acknowledge support from the Simons Center for Geometry and Physics, Stony Brook University, where this research was initiated. RM and RS would also like to thank IFT-Madrid for kind  hospitality at various stages of this project.

\appendix

\section{Some homological algebra}
\label{ap:homalgebra} 

In this appendix, we explain in more detail some of the constructions appearing in Section \ref{s:sufficient} and prove the various intermediate steps therein.

First, we would like to prove the following fact, which is needed to compute the trace of the Higgs field, like in Eq. \eqref{traceformula}:
\begin{lemma}\label{LemmaA1}
If $\alpha \in H^0(\Hom(E,\mathcal{L}_1 \otimes K_X))$, then $\tr(\tilde{i}_*(\alpha)) = \alpha \circ i$ where $\tilde{i}_*(\alpha) = (i \otimes \Id_{K_X}) \circ \alpha \in H^0(\End(E) \otimes K_X)$.
\end{lemma}

\begin{proof}
The proof consists of a local computation. First note that all the sheaves involved, namely, $\mathcal{L}_1$, $E$ and $K_X$, are locally free. Every point in $X$ thus admits an open neighbourhood $U$ over which $\mathcal{L}_1|_U \simeq \mathcal{L}_1 \otimes K_X|_U \simeq \mathcal{O}_X(U)$ and $E|_U \simeq E \otimes K_X|_U\simeq \mathcal{O}_X(U) \oplus \mathcal{O}_X(U)$.
The first terms of the extension \eqref{ext2}, when restricted to $U$, then become
\begin{equation} 
0  \xrightarrow{} \mathcal{O}_X(U) \xrightarrow{i}  \mathcal{O}_X(U) \oplus \mathcal{O}_X(U)\xrightarrow{} \cdots. 
\end{equation}
By $\mathcal{O}_X$-linearity of $i$, this means that, for all $f \in \mathcal{O}_X(U)$, $i(f) = (fh_1,fh_2)$ for some $h_1,h_2 \in \mathcal{O}_X(U)$.\footnote{Note that this Lemma equally holds in the case of the twisted sequence \eqref{extis}, where $Z$ is the locus where $h_1$ and $h_2$ vanish simultaneously.}
Now, for any $\alpha \in H^0(\Hom(E,\mathcal{L}_1 \otimes K_X))$, we have  
\begin{equation} 
\alpha|_U \in \Hom(E,\mathcal{L}_1 \otimes K_X)(U) = \Hom_{\mathcal{O}_X(U)}(\mathcal{O}_X(U) \oplus \mathcal{O}_X(U),\mathcal{O}_X(U))
\end{equation}
so that $\alpha|_U(f,g) = \alpha_1f+\alpha_2g$ for some $\alpha_1,\alpha_2 \in \mathcal{O}_X(U)$, for all $f,g \in \mathcal{O}_X(U)$. In particular,
\begin{equation} 
\alpha \circ i|_U = h_1\alpha_1 + h_2\alpha_2\,.
\end{equation}
At the same time, 
\begin{align} 
\tilde{i}_*(\alpha)|_U(f,g) &= (i \otimes \Id_{K_X}) \circ \alpha|_U(f,g)
= (i \otimes_{\mathcal{O}_X(U)} \Id_{K_X}) (\alpha_1f+\alpha_2g)\\
 &= i(\alpha_1f+\alpha_2g) = (\alpha_1h_1f + \alpha_2h_1g,\alpha_1h_2f+\alpha_2h_2g)\,. 
\end{align}
In other words,
\begin{equation} 
\tilde{i}_*(\alpha)|_U = \begin{pmatrix}
h_1\alpha_1 & h_1\alpha_2\\
h_2\alpha_1 & h_2\alpha_2
\end{pmatrix} \in \End_{\mathcal{O}_X(U)}(\mathcal{O}_X(U) \oplus \mathcal{O}_X(U)) = \End(E) \otimes K_X(U)\,,
\end{equation}
and $\tr(\tilde{i}_*(\alpha)|_U) = h_1\alpha_1 + h_2\alpha_2 = \alpha \circ i|_U$, proving the result.  \\
\end{proof}

In what follows, we fill in some gaps left in Section \ref{s:sufficient}, when discussing bundle extensions twisted by an ideal sheaf.
Let $E$ be a rank-2 vector bundle on a K\"ahler surface $S$. Suppose that $E$ is given by the extension 
\begin{equation}\label{extension-ideal}
0 \rightarrow \mathcal{L}_1 \xrightarrow{i} E \xrightarrow{q} \mathcal{L}_2 \otimes I_Z \rightarrow 0\,,
\end{equation}
where $\mathcal{L}_1$ and $\mathcal{L}_2$ are line bundles on $S$, and $Z$ is a finite set of distinct points in $S$. In other words, $E$ is a locally-free extension of $\mathcal{L}_2 \otimes I_Z$ by $\mathcal{L}_1$.

A first question to address is when such locally-free extensions exist. Suppose that $Z = \{ P_1,\dots,P_n\}$ consists of distinct (reduced) points. A locally-free extension of $\cL_2 \otimes I_Z$ by $\cL_1$ then exists if and only if every section of $\cL_1^{-1} \otimes  \cL_2\otimes K_S$ that vanishes at all but one of the $P_i$'s also vanishes at the remaining point (see \cite{Friedman}, Theorem 12). Note that this condition is vacuously satisfied when $Z$ is a single point. Moreover, if one takes sufficiently many points $\{ P_1,\dots,P_n\}$ in general position on $S$, then there will be no section of $\cL_1^{-1} \otimes  \cL_2\otimes K_S$ that vanishes at all but one of the $P_i$'s. One can therefore construct many examples of locally free extensions of $\cL_2 \otimes I_Z$ by $\cL_1$.

Let us now explain how we obtained the exact sequences \eqref{1strowideal} and \eqref{3strowideal}. These in fact correspond to the long exact cohomology sequences associated to the short exact sequences of sheaves given in the following:

\begin{proposition}\label{short-exact}
Let $E$ be a locally-free sheaf given by an extension of the form \eqref{extension-ideal}. We then have the following short exact sequences of sheaves:
\begin{equation}\label{short-exact1}
0 \rightarrow \cL_2^{-1} \otimes \cL_1 \otimes K_S \xrightarrow{q^*}
\Hom(E,\cL_1 \otimes K_S) \xrightarrow{i^*} K_S \otimes I_Z \rightarrow 0
\end{equation}
and
\begin{equation}\label{short-exact2}
0 \rightarrow K_S \xrightarrow{q^*}
\Hom(E,\cL_2 \otimes K_S\otimes I_Z ) \xrightarrow{i^*} \cL_1^{-1} \otimes \cL_2 \otimes K_S \otimes I_Z^2 \rightarrow 0\,.
\end{equation}
\end{proposition}
Indeed, by taking the long exact sequence in cohomology of \eqref{short-exact1}, we obtain \eqref{1strowideal}, where $\Lambda_3$ denotes the connecting homomorphism.
By doing the same for \eqref{short-exact2}, we obtain \eqref{3strowideal}, with $\Lambda_+$ the connecting homomorphism.

\begin{proof}[Proof of Proposition \ref{short-exact}]
We begin with some local considerations. We first note that since we are assuming that $Z$ is a finite set of distinct points, it is enough to understand what happens on open sets that contain no points in $Z$ and on open sets that only contain one point in $Z$. 
Suppose then that $Z$ consists of a single point $P$ and that $(x,y)$ are local coordinates of $S$ centered at $P$, so that $I_Z$ is generated by $x$ and $y$ on the open set $U$ on which $x$ and $y$ are defined. Then $I_Z$ fits into the following exact sequence
\begin{equation}\label{ideal-resolution}
0 \rightarrow \mathcal{O}_S \xrightarrow{i} \mathcal{O}_S \oplus \mathcal{O}_S \xrightarrow{q} I_Z \rightarrow 0\,,
\end{equation}
with $i(1) = (x,y)$ and $q(a,b) = ay-bx$ for all local sections $(a,b)$ of $\mathcal{O}_S \oplus \mathcal{O}_S$ on $U$.

Applying $\Hom(\ - \ , \mathcal{O}_S )$ to \eqref{ideal-resolution}, we obtain the exact sequence of sheaves
\begin{equation}\label{first-piece-local}
0 \rightarrow \Hom(I_Z,\mathcal{O}_S) \xrightarrow{q^*}
\Hom( \mathcal{O}_S \oplus \mathcal{O}_S, \mathcal{O}_S) \xrightarrow{i^*} \Hom( \mathcal{O}_S, \mathcal{O}_S) \rightarrow
Ext^1(I_Z,  \mathcal{O}_S) \rightarrow 0\,.
\end{equation}
Note that $\Hom(I_Z,\mathcal{O}_S) \simeq \mathcal{O}_S$ with every local homomorphism $I_Z \rightarrow \mathcal{O}_S$ corresponding to multiplication by an element in $\mathcal{O}_S$ (see \cite{Friedman}, Lemma 7). In addition, we have the obvious isomorphisms $\Hom( \mathcal{O}_S \oplus \mathcal{O}_S, \mathcal{O}_S)  \simeq \mathcal{O}_S \oplus \mathcal{O}_S$ and
$\Hom(\mathcal{O}_S,\mathcal{O}_S) \simeq \mathcal{O}_S$.
Then, for every local homomorphism $(a,b):\mathcal{O}_S \oplus \mathcal{O}_S\to\cO_S$ represented by the pair $(a,b)\in\mathcal{O}_S \oplus \mathcal{O}_S$, we have $i^*((a,b))(1) = (a,b)\circ i(1)  = (a,b)\, {x\choose y} = ax+by$, implying that the map $i^*((a,b))$ is multiplication by $ax+by$ in $\mathcal{O}_S$. Nevertheless, $ax+by \in I_Z$, meaning that $i^*(\Hom( \mathcal{O}_S \oplus \mathcal{O}_S, \mathcal{O}_S)) = \Hom(\mathcal{O}_S,I_Z) \simeq I_Z$ and $Ext^1(I_Z,\mathcal{O}_S) \simeq \mathcal{O}_S/I_Z$. The first two terms of \eqref{first-piece-local} thus fit into the short exact sequence
\begin{equation}\label{first-piece-local-exact}
0 \rightarrow \Hom(I_Z,\mathcal{O}_S) \xrightarrow{q^*}
\Hom( \mathcal{O}_S \oplus \mathcal{O}_S, \mathcal{O}_S) \xrightarrow{i^*} \Hom( \mathcal{O}_S, I_Z) \rightarrow 0\,.
\end{equation}

Applying $\Hom(\ - \ , I_Z)$ to \eqref{ideal-resolution}, instead, we obtain the exact sequence of sheaves
\begin{equation}\label{local-picture-ideal}
0 \rightarrow \Hom(I_Z,I_Z) \xrightarrow{q^*}
\Hom( \mathcal{O}_S \oplus \mathcal{O}_S, I_Z) \xrightarrow{i^*} \Hom( \mathcal{O}_S, I_Z) \rightarrow
Ext^1(I_Z, I_Z) \rightarrow 0\,.
\end{equation}
We again have
$\Hom(I_Z,I_Z) \simeq \mathcal{O}_S$ with every local homomorphism $I_Z \rightarrow I_Z$ corresponding to multipication by an element in $\mathcal{O}_S$ (see \cite{Friedman}, Lemma 7). Furthermore, $\Hom(\mathcal{O}_S,I_Z) \simeq I_Z$ with every local homomorphism $\mathcal{O}_S\rightarrow I_Z$ corresponding to multiplication by an element in $I_Z$. Finally, $\Hom( \mathcal{O}_S \oplus \mathcal{O}_S, I_Z)  \simeq I_Z \oplus I_Z$ since every local homomorphism $\mathcal{O}_S \oplus \mathcal{O}_S \rightarrow I_Z$ is of the form $(f,g) \mapsto af+bg$ with $(a,b) \in I_Z\oplus I_Z$.
 We again have $i^*((a,b))(1) = (a,b)\circ i(1)  = (a,b)\, {x\choose y} = ax+by$, implying that the map $i^*((a,b))$ is multiplication by $ax+by$ in $\mathcal{O}_S$, but this time $ax+by \in I^2_Z$. Hence, $i^*(\Hom( \mathcal{O}_S \oplus \mathcal{O}_S, I_Z)) = \Hom(\mathcal{O}_S,I_Z^2) \simeq I_Z^2$ and $Ext^1(I_Z,I_Z) \simeq I_Z/I_Z^2$. Consequently, the first two terms of \eqref{local-picture-ideal} split into the short exact sequence
\begin{equation}\label{first-piece-local-exact-ideal}
0 \rightarrow \Hom(I_Z,I_Z) \xrightarrow{q^*}
\Hom( \mathcal{O}_S \oplus \mathcal{O}_S,I_Z) \xrightarrow{i^*} \Hom( \mathcal{O}_S, I_Z^2) \rightarrow 0\,.
\end{equation}

Let us now derive \eqref{short-exact1}.
Applying $\Hom(\ - \ ,\cL_1)$ to \eqref{extension-ideal}, we obtain the exact sequence of sheaves
\begin{equation}\label{first-piece}
0 \rightarrow \Hom(\cL_2 \otimes I_Z, \cL_1) \xrightarrow{q^*}
\Hom(E,\cL_1) \xrightarrow{i^*} \Hom(\cL_1,\cL_1) \rightarrow
Ext^1(\cL_2 \otimes I_Z, \cL_1) \rightarrow 0 \,.
\end{equation}
In light of the above discussion, we see that 
\begin{equation} 
i^*(\Hom(E,\cL_1)) = \Hom(\cL_1,\cL_1 \otimes I_Z)\,,
\end{equation}
giving rise to the short exact sequence
\begin{equation}\label{short-exact-hom1}
0 \rightarrow \Hom(\cL_2 \otimes I_Z, \cL_1) \xrightarrow{q^*}
\Hom(E,\cL_1) \xrightarrow{i^*} \Hom(\cL_1,\cL_1 \otimes I_Z) \rightarrow 0\,.
\end{equation}
To understand this sequence better, recall that for any two sheaves $\mathcal{E}_1, \mathcal{E}_2$, we have $\Hom(\mathcal{E}_1,\mathcal{E}_2) \simeq \mathcal{E}_1^* \otimes \mathcal{E}_2$. Moreover, from the above local discussion,
\begin{equation} 
\Hom(\cL_2 \otimes I_Z, \cL_1) \simeq  \Hom(\cL_2, \cL_1) \simeq \cL_2^{-1} \otimes \cL_1
\end{equation}
and
\begin{equation}  
\Hom(\cL_1,\cL_1 \otimes I_Z) \simeq \cL_1^{-1} \otimes \cL_1 \otimes I_Z \simeq I_Z\,,
\end{equation}
where we used that the dual of a line bundle coincides with its inverse.
The short exact sequence \eqref{short-exact-hom1} can therefore be written as
\begin{equation}
0 \rightarrow  \cL_2^{-1} \otimes \cL_1 \xrightarrow{q^*}
\Hom(E,\cL_1) \xrightarrow{i^*}  I_Z \rightarrow 0\,,
\end{equation}
which tensored by $K_S$ gives \eqref{short-exact1}, since
$\Hom(E,\cL_1) \otimes K_S \simeq (E^* \otimes \cL_1) \otimes K_S \simeq \Hom(E,\cL_1 \otimes K_S)$.

We now turn to \eqref{short-exact2}. Applying $\Hom(\ - \ , \cL_2 \otimes I_Z)$ to \eqref{extension-ideal}, we get the exact sequence {\small
\begin{equation}
0 \rightarrow \Hom(\cL_2 \otimes I_Z, \cL_2 \otimes I_Z) \xrightarrow{q^*}
\Hom(E,\cL_2 \otimes I_Z) \xrightarrow{i^*} \Hom(\cL_1,\cL_2 \otimes I_Z)
\rightarrow Ext^1(\cL_2 \otimes I_Z, \cL_2 \otimes I_Z)\rightarrow 0\,.
\end{equation}}
This time, 
\begin{equation} 
\Hom(\cL_2 \otimes I_Z, \cL_2 \otimes I_Z) \simeq  \Hom(\cL_2, \cL_2) \simeq \mathcal{O}_S
\end{equation}
and
\begin{equation} 
i^*(\Hom(E,\cL_2 \otimes I_Z)) = \Hom(\cL_1,\cL_2 \otimes I_Z^2)  \simeq \cL_1^{-1} \otimes \cL_2 \otimes I_Z^2\,,
\end{equation}
giving us the short exact sequence
\begin{equation}
0 \rightarrow \mathcal{O}_S \xrightarrow{q^*}
\Hom(E,\cL_2 \otimes I_Z) \xrightarrow{i^*} \cL_1^{-1} \otimes \cL_2 \otimes I_Z^2\rightarrow 0\,,
\end{equation}
which tensored by $K_S$ yields \eqref{short-exact2}.
\end{proof}

\section{Topologically non-split bundles}
\label{TopNonSp}

In this appendix, we will describe the topological features of a rank-$2$ vector bundle on a K\"ahler surface $S$, defined by the twisted extension \eqref{extis}. As we will see, twisting the sequence by an ideal sheaf gives us a simple way to construct a class of rank-$2$ vector bundles which may be topologically disconnected from a split form of direct sum of two line bundles.

Let $Z$ be a set of distinct points in $S$,\footnote{In the context of Type IIB Orientifold (or F-theory) compactifications to four dimensions, $Z$ corresponds to a set of D3-branes extended over the external space and completely localised in the internal part of the 7-brane worldvolume. Their orientation is fixed by supersymmetry.} with $j:Z\to S$ the corresponding holomorphic embedding. The Grothendieck-Riemann-Roch formula says:
\begin{equation}
{\rm ch}(j_*\mathcal{O}_Z){\rm td}(S)=j_{\#}({\rm ch}(\mathcal{O}_Z){\rm td}(Z))\,,
\end{equation}
where ``ch'' and ``td'' denote the total Chern character and the total Todd class respectively, $\mathcal{O}_Z$ is the structure sheaf of $Z$, $j_*$ is the push-forward map (naturally acting on sheaves), and finally $j_{\#}$ is the push-forward in cohomology, meaning that it consists of taking first a Poincar\'e duality, then the natural push-forward of homology classes, and then another Poincar\'e duality\footnote{Effectively this set of operations sends cohomology classes of $Z$ to cohomology classes of $S$ by simply wedging them with the class dual to $Z$ in $S$.}. The l.h.s.~of the above formula expands to
\begin{equation}
\left[{\rm rk}(j_*\mathcal{O}_Z)+c_1(j_*\mathcal{O}_Z)+\frac{1}{2}c_1^2(j_*\mathcal{O}_Z)-c_2(j_*\mathcal{O}_Z)\right]\left[1+\frac{1}{2}c_1(S)+\frac{1}{12}(c_1^2(S)+c_2(S))\right]\,,\nonumber
\end{equation}
whereas the r.h.s.~simply reads
\begin{equation}
{\rm PD}_S(Z)\,,\nonumber
\end{equation}
which is the four-class Poincar\'e dual to $Z$ in $S$. Equating order by order, we get useful topological information about the skyscraper sheaf $j_*\mathcal{O}_Z$:
\begin{equation}
{\rm rk}(j_*\mathcal{O}_Z) = c_1(j_*\mathcal{O}_Z) = 0\,, \qquad {\rm and}\qquad c_2(j_*\mathcal{O}_Z)= - \,{\rm PD}_S(Z)\,.
\end{equation}
Since $j_*\mathcal{O}_Z$ is by definition the cokernel of the embedding $I_Z\to \mathcal{O}_S$, with $\mathcal{O}_S$ the structure sheaf of $S$, we trivially get:
\begin{equation}
{\rm rk}(I_Z) = 1\,,\qquad c_1(I_Z) = 0\,, \qquad {\rm and}\qquad c_2(I_Z)=  {\rm PD}_S(Z)\,.
\end{equation}

Using the properties of the Chern character on exact sequences, it is immediate to verify that both the rank and the first Chern class of $E$ in \eqref{extis} are insensitive to the presence of $I_Z$, namely rk$(E)=2$ and $c_1(E)=c_1(\mathcal{L}_1)+c_1(\mathcal{L}_2)$. The second Chern number of $E$, however, crucially receives a shift determined by the number of points in $Z$, i.e.
\begin{equation}\label{c2E}
c_2(E)=c_1(\mathcal{L}_1)c_1(\mathcal{L}_2)+{\rm PD}_S(Z)\,.
\end{equation}

It is now easy to design an explicit example where a vector bundle $E$ given by \eqref{extis} is topologically obstructed to splitting into a sum of two line bundles. Consider, for instance, a surface $S$ cut out by a \emph{generic} polynomial of degree $7$ in $\mathbb{P}^3$.\footnote{We may physically motivate this in the context of F-theory, where $\mathbb{P}^3$ plays the role of base of the elliptic fibration.} This surface is simply connected and has $20$ independent holomorphic deformations. Moreover, even though $S$ admits $147$ independent harmonic $(1,1)$-forms, the genericity of the degree-7 polynomial implies that all of them but one be non-integral.\footnote{This is a trivial consequence of the N\"other-Lefschetz theorem, which in this case says that the restriction ${\rm Pic}(\mathbb{P}^3)\to{\rm Pic}(S)$ is an isomorphism.} This means that we can only play with the hyperplane class to construct line bundles on $S$. Take for simplicity $\mathcal{L}_2^{-1}\simeq \mathcal{L}_1\equiv \mathcal{L}$, and choose $c_1(\mathcal{L})=-H$, where $H$ is the hyperplane class of $\mathbb{P}^3$. This guarantees that $E$ is stable. Using that $c_1(S)=-3H$, we find that the number of points in $Z$ is $21$, and therefore, by Eq. \eqref{c2E}, we have $\int c_2(E)=14$. However, any traceless sum of two line bundles on this surface has a negative second Chern number, thus proving that the vector bundle $E$ so constructed does not admit a split form in its moduli space of complex structures.

\section{Cohomological computations}
\label{ap:cohomology}

In this appendix, we provide a few cohomological computations that are needed in Section \ref{s:examples}, in which we construct examples of Higgs bundles over properly elliptic surfaces.

Suppose that $S = C \times T$ where $C$ is a curve of genus $\geq 2$ and $T$ is an elliptic curve. Then, $S$ is a properly elliptic surface with trivial elliptic fibration given by projection onto the first factor $\pi := {\rm pr}_1: S \rightarrow C$. Note that $K_S \simeq \pi^*K_C$. Suppose that $C$ is {\em not} hyperelliptic. This means, in particular, that $g \geq 3$. In this case, we have:

\begin{proposition}\label{cohomology-on-surface}
	For any $P_0 \in C$, set $T_0 := \{ P_0 \} \times T$ and 
	\[K_S(nT_0) := \mathcal{O}_S(nT_0) \otimes K_S\] 
	for all $n \in \mathbb{Z}$. Then:
	\begin{enumerate}

		\item $h^0(S,K_S(sT_0)) 
		= h^1(S,K_S(sT_0)) = g+s-1$ for all $s>0$.
		
		\item $h^0(S,K_S(-sT_0)) = g-s$ and $h^1(S,K_S(-sT_0)) = g-s+1$ for $s=0,1,2$.
		
		\item $g-s \leq h^0(S,K_S(-sT_0)) \leq g-2$ and $g-s+1 \leq h^1(S,K_S(-sT_0)) \leq g+s-3$ for all $s \geq 3$.
	\end{enumerate}
\end{proposition}

Before proving the proposition, we need to first compute some coholomogy groups over the base curve $C$.

\begin{lemma}
	For any fixed $P_0 \in C$, we have:
	\begin{enumerate}
		\item  $h^0(C,\mathcal{O}_C(sP_0)) = 1$ and $h^1(C,\mathcal{O}_C(sP_0)) = g-s$ for all $s = 0,1,2$.
		
		\item $1 \leq h^0(C,\mathcal{O}_C(sP_0)) \leq s-1$ and
		$g-s \leq h^1(C,\mathcal{O}_C(sP_0)) \leq g-2$ for all $s \geq 3$. 
	\end{enumerate}
\end{lemma}
\begin{proof}
	First note that, by Riemann-Roch, 
	\begin{equation} 
	h^1(C,\mathcal{O}_C(sP_0)) = h^0(C,\mathcal{O}_C(sP_0)) + g - (1+s)
	\end{equation}
	for all $s$. For $s=0,1,2$, it is therefore enough to check that $h^0(C,\mathcal{O}_C(sP_0)) = 1$. If $s=0$, the statement is clear. If $s\geq 1$, recall that
	\begin{equation} 
	H^0(C,\mathcal{O}_C(sP_0)) = \{ f \in \mathcal{M}(C) : (f) \geq -sP_0\}\,,
	\end{equation}
	where $\mathcal{M}(C)$ is the set of meromorphic functions on $C$ and $(f)$ is the divisor of the meromorphic function $f$. In particular, if $f \in  H^0(C,\mathcal{O}_C(sP_0))$ is not constant, then it must have a single pole of order $\leq s$ at $P_0$. If $s=1$, this single pole must have order 1, implying that $f:C \rightarrow \mathbb{P}^1$ is a degree 1 map. In other words, if $h^0(C,\mathcal{O}(P_0)) \geq 2$, then $C$ is isomorphic to $\mathbb{P}^1$, which is impossible since we are assuming that
	$g \geq 3$. Finally, when $s=2$, if $h^0(C,\mathcal{O}_C(2P_0)) \geq 2$, this would imply that there is a non-constant $f \in \mathcal{M}(C)$ with $(f) \geq -2P_0$ so that $f$ must have a pole of order $\leq 2$ at $P_0$. As in the previous case, this pole cannot have order 1 because $C$ has genus $\geq 3$. Therefore, $f$ must have a pole of order 2 at $P_0$, implying that $f: C \rightarrow \mathbb{P}^1$ has degree two. But this would mean that $C$ is hyperelliptic, which again contradicts our assumption and proves 1.. 
	
	The statement of 2. follows from the exact sequence
	\begin{equation} 
	0 \rightarrow \mathcal{O}_C(nP_0) \rightarrow \mathcal{O}_C((n+1)P_0) \rightarrow \mathcal{O}_{P_0} \rightarrow 0
	\end{equation}
	with $n \in \mathbb{Z}$. Indeed, if we take $n=2$, we then have
	\begin{equation} 
	0 \rightarrow \mathcal{O}_C(2P_0) \rightarrow \mathcal{O}_C(3P_0) \rightarrow \mathcal{O}_{P_0} \rightarrow 0\,.
	\end{equation}
	Taking the long cohomology sequence, we obtain
	\[0 \rightarrow H^0(C,\mathcal{O}_C(2P_0)) \rightarrow H^0(C,\mathcal{O}_C(3P_0)) \rightarrow
	H^0(C,\mathcal{O}_{P_0}) \rightarrow H^1(C,\mathcal{O}_C(2P_0)) \rightarrow\]
	\begin{equation} 
	H^1(C,\mathcal{O}_C(3P_0)) \rightarrow H^1(C,\mathcal{O}_{P_0}) \rightarrow 0\,. 
	\end{equation}
	However, $h^0(C,\mathcal{O}_{P_0}) = 1$ and $h^1(C,\mathcal{O}_{P_0}) = 0$ since $\mathcal{O}_{P_0}$ is a torsion sheaf supported on a single point. Moreover, $h^0(C,\mathcal{O}_C(2P_0)) = 1$
	and $h^1(C,\mathcal{O}_C(2P_0)) = g-2$ by 1.. The long exact sequence on cohomology thus reduces to
	\[0 \rightarrow \mathbb{C} \rightarrow H^0(C,\mathcal{O}_C(3P_0)) \rightarrow
	\mathbb{C} \rightarrow \mathbb{C}^{g-2} \rightarrow H^1(C,\mathcal{O}_C(3P_0)) \rightarrow 0\,. \]
	We then see that $1 \leq h^0(C,\mathcal{O}_C(3P_0)) \leq 2$ and $g-3 \leq h^1(C,\mathcal{O}_C(3P_0)) \leq g-2$, implying that the statement of 2. holds for $s=3$. The result for $s \geq 4$ follows by induction.
\end{proof}

Using this lemma, we can now compute the cohomology groups on $S$.

\begin{proof}[Proof of Proposition \ref{cohomology-on-surface}]
	First note that since $S = C \times T$ and $\pi$ is projection onto the first factor, we have $\pi_*(\mathcal{O}_S) = R^1\pi_*(\mathcal{O}_S) = \mathcal{O}_C$. Moreover, 
	$\mathcal{O}_S(nT_0) = \pi^*(\mathcal{O}_C(nP_0))$ for all $n \in \mathbb{Z}$ and, since $K_S = \pi^* K_C$, we have 
	\begin{equation} 
	K_S(nT_0) = \pi^*(\mathcal{O}_C(nP_0)) \otimes K_S = \pi^*(K_C(nP_0))
	\end{equation} 
	for all $n$
	(where $K_C(nP_0) := \mathcal{O}_C(nP_0) \otimes K_C$). Therefore, by the projection formula,
	\begin{equation} 
	\pi_*(K_S(nT_0)) = R^1\pi_*(K_S(nT_0)) = K_C(nP_0)
	\end{equation}
	for all $n$. Consequently,
	\begin{equation} 
	H^0(S,K_S(nT_0)) = H^0(C,\pi_*(K_S(nT_0))) = H^0(C,K_C(nP_0)) 
	\end{equation}
	and, by the Leray spectral sequence,
	\begin{align} 
	H^1(S,K_S(nT_0)) &= H^0(C,R^1\pi_*(K_S(nT_0))) \oplus H^1(C,\pi_*(K_S(nT_0))) \\
	&= H^0(C,K_C(nP_0)) \oplus H^1(C,K_C(nP_0))\,.
	\end{align}
	Moreover, by Serre duality,
	\begin{equation} 
	h^0(C,K_C(nP_0)) = h^1(C,\mathcal{O}_C(-nP_0)) 
	\end{equation}
	and
	\begin{equation} 
	h^1(C,K_C(nP_0)) = h^0(C,\mathcal{O}_C(-nP_0))\,. 
	\end{equation}
	Hence,
	\begin{equation} 
	h^0(S,K_S(nT_0)) = h^1(C,\mathcal{O}_C(-nP_0)) = h^0(C,\mathcal{O}_C(-nP_0)) + g + n -1\,,
	\end{equation}
	where the second equality follows from Riemann-Roch,
	and
	\begin{equation} 
	h^1(S,K_S(nT_0)) =  h^1(C,O_C(-nP_0)) + h^0(C,\mathcal{O}_C(-nP_0))
	= 2h^0(C,\mathcal{O}_C(-nP_0)) + g + n -1\,. 
	\end{equation}
	
	If $n=s>0$, then $h^0(C,\mathcal{O}_C(-sP_0)) = 0$, implying that
	\begin{equation} 
	h^0(S,K_S(sT_0)) = h^1(S,K_S(sT_0)) = g + s -1
	\end{equation}
	for all $s>0$.
	On the other hand, if $n = -s$ with $s=0,1,2$, then $h^0(C,\mathcal{O}_C(sP_0)) = 1$, 
	implying that
	\begin{equation} 
	h^0(S,K_S(-sT_0)) = g - s
	\end{equation}
	and
	\begin{equation} 
	h^1(S,K_S(-sT_0)) = g - s +1\,.
	\end{equation}
	
	Finally, assume that $n = -s$ with $s \geq 3$. Then,
	\begin{equation}
	1 \leq h^0(C,\mathcal{O}_C(sP_0)) \leq s-1\,,
	\end{equation}
	implying that
	\begin{equation} 
	g-s \leq  h^0(S,K_S(-sT_0)) \leq g-2
	\end{equation}
	and
	\begin{equation} 
	g-s+1 \leq  h^1(S,K_S(-sT_0)) \leq g+s-3\,,
	\end{equation}
	proving 3.
\end{proof}

\begin{remark}
	The proposition tells us, in particular, that
	\begin{equation} 
	h^1(S,K_S(-2T_0))) = g-1 < g,
	\end{equation}
	which we used to prove the existence of traceless Higgs fields with $\tilde{i}_*(\hat{v}) \neq 0$ on $S$ in section \ref{s:examples}.	
\end{remark}

\end{document}